\documentclass[aps,pra,english,10pt,a4,nofootinbib,twocolumn,notitlepage,superscriptaddress]{revtex4-1}
\usepackage {graphicx}
\usepackage{color}
\usepackage{amsmath, amssymb,amsthm}
\usepackage{mathtools}
\usepackage{multirow}

\usepackage{hyperref,graphicx, amsmath, amssymb, amsthm, color, bm, bbm,dsfont}
\usepackage{cleveref}

\theoremstyle{plain}
\newtheorem{thm}{Theorem}
\crefname{thm}{theorem}{theorems}

\crefname{prop}{proposition}{propositions}
\newtheorem{cor}[thm]{Corollary}
\theoremstyle{definition}

\definecolor{nblue}{rgb}{0.2,0.2,0.7}
\definecolor{ngreen}{rgb}{0.1,0.5,0.1}
\definecolor{nred}{rgb}{0.8,0.2,0.2}
\definecolor{nblack}{rgb}{0,0,0}

\DeclareMathOperator{\tr}{\mathrm{Tr}}
\usepackage{wasysym}

\newcommand{\ds}[1]{\mathds{#1}}

\newcommand{\mc}[1]{\mathcal{#1}}

\newcommand{\bc}[1]{\bm{\mathcal{#1}}}
\newcommand{\mbb}[1]{\mathbb{#1}}

\newcommand{\la}{\langle}
\newcommand{\ra}{\rangle}
\newcommand{\lv}{\lvert}

\newcommand{\rv}{\rvert}

\newcommand{\dla}{\langle\!\langle}
\newcommand{\dra}{\rangle\!\rangle}

\newcommand{\ct}{^{\dagger}}

\newcommand{\hidden}[1]{}

\begin{document}
	
\title{{Bounding the average gate fidelity of composite channels using the unitarity}}
	\author{Arnaud Carignan-Dugas}
	\affiliation{Institute for Quantum Computing and the Department of Applied
		Mathematics, University of Waterloo, Waterloo, Ontario N2L 3G1, Canada}
	\author{Joel J. Wallman}
	\affiliation{Institute for Quantum Computing and the Department of Applied
		Mathematics, University of Waterloo, Waterloo, Ontario N2L 3G1, Canada}
	\author{Joseph Emerson}
	\affiliation{Institute for Quantum Computing and the Department of Applied
		Mathematics, University of Waterloo, Waterloo, Ontario N2L 3G1, Canada}
	\affiliation{Canadian Institute for Advanced Research, Toronto, Ontario M5G 1Z8, Canada}
\begin{abstract}
There is currently a significant need for robust and efficient methods for characterizing quantum devices. 
While there has been significant progress in this direction, there remains a crucial need to precisely determine the strength and type of errors  on 
individual gate operations, in order to assess and improve control as well as reliably bound the total error in a quantum circuit given some partial information about the errors on the components. In this work, we first
provide an optimal bound on the total fidelity of a circuit in terms of 
component fidelities, 
which can be efficiently experimentally estimated via randomized 
benchmarking. We then derive a tighter bound that applies under additional information about 
the coherence of the error, namely, the unitarity, which can also be estimated via a related experimental protocol. This 
improved bound smoothly interpolates between the worst-case quadratic and 
best-case linear scaling for composite error channels. As an application we show how our analysis substantially
improves the achievable precision on estimates of the infidelities of individual gates under 
interleaved randomized benchmarking, enabling greater precision for current experimental methods to assess and tune-up control over quantum gate operations.
\end{abstract}

\maketitle
\section{Introduction}

The output of a quantum computer will only be reliable if the total error in 
the whole computation is sufficiently small. This can be rigorously guaranteed if the 
error on the individual components (i.e., preparations, measurements and gate operations)
is sufficiently small compared to the length of the computation.   
{A very common} experimental practice \cite{ Gaebler2012,Corcoles2013,Kelly2014,Barends2014,Xia2015,Muhonen2015,Tarlton,Casparis2016,McKay2016,Sheldon2016,Takita2016,McKay2017} for estimating 
errors on gate operations
 is  randomized benchmarking (RB) of Clifford operations \cite{Emerson2005,DCEL2006}.
The experimentally measured infidelities under RB experiments have very recently been shown to give a very precise estimate of 
the average gate fidelity (hereafter simply the fidelity) of an 
error channel to the identity,
\begin{align}
F(\mc{E}) := \int {\rm d}\psi \la \psi|\mc{E}(|\psi\ra\!\la \psi|)|\psi\ra
\end{align}
under very robust and experimentally realistic conditions\cite{Magesan2011, Magesan2012b, Wallman2015, Wallman2015b, Wallman2016, Proctor2017, Wallman2017, Dugas2018,Harper2019,Dirkse2019}, 
{when expressed in a physical operational gauge \cite{Proctor2017,Wallman2017,Merkel2018,Dugas2018}}\footnote{{In Dugas et al, the physicality of the gauge is proven for $d=2$, and conjectured otherwise.}}, {resolving the concern (that RB did not reliably measure a physically meaningful fidelity) raised in \cite{Proctor2017v1,Qi2018}.} 

An important practical application of RB is interleaved RB (IRB) \cite{Magesan2012b},  a now-standard approach 
 for estimating infidelities on individual gates \cite{Gaebler2012,Corcoles2013,Kelly2014,Barends2014,Veldhorst2014,Muhonen2015,Veldhorst2015,Barends2015,Casparis2016,McKay2016,Sheldon2016,Takita2016,Takeda2016,McKay2017,McKay2017b,Nichol2017,Chan2018,Caldwell2018,Wang2018,Wang2018b,Yoneda2018,Zhang2018}, including gates that collectively generate universality \cite{Dugas2015,Cross2016,Harper2017,Proctor2018}. However this approach is subject to a systematic error that can significantly limit the precision of the estimate and often goes unreported - a problem which we address below.
As noted above, the average gate fidelity gives only very limited information about the error and error 
channels with the same fidelity on the component gate operations can lead to dramatically different total error for a
circuit composed from these gate operations. For example, the infidelity $r(\mc{E}) = 1 - F(\mc{E})$ grows linearly 
in the number of gates under purely stochastic errors (that is, errors that 
can be modeled by classical probabilities over different Pauli operators) and grows
quadratically under purely unitary errors (that is, coherent errors due to small calibrations that are common in quantum control) in the limit of small 
infidelities~\cite{Sheldon2016}. However, realistic experimental errors are 
neither purely stochastic nor purely unitary, but rather some combination 
of the two. {To adequately characterize quantum circuits, which are the result of multiple noisy operations, it is crucial to understand (and bound) how errors can accumulate given an intermediate level of coherence.} In this paper, we study 
{the impact of coherence on the fidelity of circuit constructions.}
An important {application} from our work is to provide a dramatic improvement to the achievable precision of IRB, enabling significantly more reliable experimental methods for assessing and tuning-up the individual gate operations required  for quantum computing and other applications. 

This paper is organized as follows. We first obtain strictly optimal 
upper- and lower-bounds on the total infidelity of the circuit for all parameter regimes when 
only the infidelities of the components are known. These bounds are saturated 
by unitary channels and so grow quadratically with the number of gates. 
Moreover, because our bounds are saturated, they cannot be improved without 
further knowledge about the errors. 
Because the worst-case growth of the infidelity is achieved by purely unitary 
channels, intuitively, quantifying how far an error channel is from purely unitary error 
should enable an improved bound. One such quantity is the unitarity.
Thus our second contribution in this work is a proof  
that the unitarity 
\begin{align}\label{def:unitarity}
u(\mc{E}) := \frac{d}{d-1}\int {\rm d}\psi 
\tr\mc{E}(\psi-\tfrac{1}{d}\ds{I})^2,
\end{align}
of the components, which can be estimated using a variant of RB~\cite{Wallman2015,Dugas2018gcfgen,Dirkse2019} (URB),
can be used to obtain a tighter bound on the total infidelity. This information enables a smooth interpolation between the 
quadratic growth of purely unitary errors and the linear scaling of purely 
stochastic errors. {Including the unitarity to characterize 
	circuits allows to quantitatively reason about an often omitted statement:
	elementary operations with low infidelity and highly coherent errors can rapidly compose to a worse
	circuit than a sequence of elementary operations with moderate infidelity but highly stochastic errors. Our bounds implicitly quantify how fast this can happen given the infidelity and unitarity of individual components.} Our third contribution, noted above, {goes the other way: from a composite error $\mc E_h \circ \mc E$, we bound the fidelity of one of its component $\mc E_h$. We demonstrate} an immediate practical application of this result by providing a dramatically improved  bound on the  accuracy of the estimates of gate infidelities under interleaved RB~\cite{Magesan2012b}. {This is done by substituting the estimate of the 
effective depolarizing constant of 
the individual interleaved gate $\hat{p}= p_{\text{IRB}}/p_{\text{RB}}$ by
$\hat{p}= p_{\text{IRB}}p_{\text{RB}}/p_{\text{URB}}$, which 
requires a unitarity RB (URB) experiment. {In the experiments reported in \cite{Xue2018,Yang2018}, our estimator is used to rigorously bound the 
	infidelity of individual quantum gates via \cref{eq:uni_bound}.}

\section{Noisy quantum processes}\label{app:representation}

Markovian quantum processes can be described by completely-positive and 
trace-preserving (CPTP) linear maps $\mc{E}:\mbb{D}_d\to\mbb{D}_d$ where 
$\mbb{D}_d$ is the set of density operators acting on $\mbb{C}^d$, that is, the 
set of positive-semidefinite operators with unit trace. We denote quantum 
channels using single calligraphic capital Roman letters and the composition of 
channels by multiplication for brevity, so that $\mc{A}\mc{B}(\rho) = 
\mc{A}[\mc{B}(\rho)]$. We also denote the composition of $m$ channels 
$\mc{E}_1,\ldots,\mc{E}_m$ by $\mc{E}_{1:m} = \mc{E}_1\ldots\mc{E}_m$.

Abstract quantum channels can be represented in many ways. In this paper, we 
will use the Kraus operator, $\chi$-matrix and the Liouville (or transfer 
matrix) representations. The Kraus operator and $\chi$-matrix representations 
of a quantum channel $\mc{E}$ are
\begin{align}
	\mc{E}(\rho) 
	= \sum_j A_j \rho A_j\ct
	= d\sum_{k,l\in\mbb{Z}_{d^2}} \chi_{kl}^{\mc{E}}B_k\rho B_l\ct 
\end{align}
respectively, where the $A_j$ are the Kraus operators, $\mbb{Z}_{d^2} = 
\{0,\ldots,d^2-1\}$ and $\mc{B}=\{B_0=\ds{I}_d/\sqrt{d},B_1,\ldots,B_{d^2-1}\}$ 
is a trace-orthonormal basis of $\mbb{C}^{d\times d}$ satisfying $\la 
B_j,B_k\ra :=\tr B_j\ct B_k = \delta_{j,k}$. Note that we include the 
dimensional factor in the definition of the $\chi$-matrix to be consistent with 
the standard construction in terms of unnormalized Pauli operators.

The Kraus operators can be expanded as $A_j = \sum_{k\in\mbb{Z}_{d^2}} \la 
B_k,A_j\ra B_k$ relative to $\mc{B}$. Making use of the phase freedom in the 
Kraus operators (that is, $A_j \to e^{i\theta_j}A_j$ gives the same quantum 
channel), we can set $\la B_0,A_j\ra \geq 0$ for all $j$. We can then expand 
the Kraus operators as
\begin{align}\label{eq:Kraus_expansion}
A_j 
=\sqrt{a_j d} 
\Big(\cos(\alpha_j)B_0+\sin(\alpha_j)\vec{v}_j\cdot\vec{\mc{B}}\Big)
\end{align}
where $a_j d = \la A_j, A_j\ra$, $\vec{\mc{B}} = (B_1,\ldots,B_{d^2-1})$, 
$\vec{v}_j\in\mbb{C}^{d^2-1}$ with $\|\vec{v}_j\|_2=1$, and $\alpha_j$ can be 
chosen to be in $[0,\tfrac{\pi}{2}]$ by incorporating any phase into 
$\vec{v}_j$. Substituting this expansion into the Kraus operator decomposition 
and equating coefficients with the $\chi$-matrix representation gives
\begin{align}
	\chi_{kl}^{\mc{E}} 
	= \frac{1}{d}\sum_j \la B_k,A_j\ra \!\la A_j,B_l\ra ,
\end{align}
and, in particular,
\begin{align}\label{eq:Kraus_to_chi00}
	\chi_{00}^{\mc{E}} 
	= \frac{1}{d^2}\sum_j |\tr A_j|^2 
	= \sum_j a_j \cos^2(\alpha_j).
\end{align}
Applying the trace-preserving constraint with ${\la B_j, B_k\ra = 
\delta_{j,k}}$ gives
\begin{align}\label{eq:Kraus_coefficients_TP}
	1= \frac{1}{d}\tr \sum_j A_j\ct A_j = \sum_j a_j,
\end{align}
which then implies
\begin{align}\label{eq:Kraus_to_rchi00}
	1-\chi_{00}^{\mc{E}} 
	= \sum_j a_j \sin^2(\alpha_j).
\end{align}

Alternatively, density matrices $\rho$ and effects $E$ (elements of 
positive-operator-valued measures) can be expanded with respect to $\mc{B}$ as 
$\rho = \sum_j \la B_j,\rho\ra B_j$ and $E = \sum_j \la B_j,E\ra B_j$. The 
Liouville representations of $\rho$ and $E$ are the column vector $|\rho\dra$ 
and row vector $\dla E| = |E\dra \ct$ of the corresponding expansion 
coefficients. The Born rule is then $\la E,\rho\ra =\dla E|\rho \dra$. The 
Liouville representation of a channel $\mc{E}$ is the unique matrix $\bc{E}$ 
such that $\bc{E}|\rho\dra = |\mc{E}(\rho)\dra$, which can be written as 
$\bc{E} = \sum_j |\mc{E}(B_j)\dra\!\dla B_j|$. With $B_0 = \ds{I}_d/\sqrt{d}$, 
the Liouville representation of any CPTP map can be expressed in block form as
\begin{align}\label{eq:block}
	\bc{E} 
	= \left(\begin{array}{cc} 1 & 0 \\ \bc{E}_{\rm n} & \bc{E}_{\rm u} 
	\end{array}\right)
\end{align} 
where $\bc{E}_{\rm n}\in\mbb{C}^{d^2-1\time 1}$ is the non-unital vector and 
$\bc{E}_{\rm u}\in\mbb{C}^{d^2-1\times d^2-1}$ is the unital block. The 
unitarity and {effective depolarizing constant} can be written as
\begin{align}\label{eq:Liouville_parameters}
	u(\mc{E}) &= \frac{\tr\bc{E}_{\rm u}\ct \bc{E}_{\rm u}}{d^2-1}= \frac{\| \bc{E}_{\rm u}\|_F^2}{d^2-1}\notag\\
	p(\mc{E}) &= \frac{\tr\bc{E}_{\rm u}}{d^2-1}
\end{align}
with respect to the Liouville representation~\cite{Wallman2015,Kimmel2014}. 

The {effective depolarizing constant} $p(\mc{E})$ and $\chi_{00}$ are linear functions of the 
fidelity that can be more convenient to work with. The relations between the 
various linear functions of the fidelity used in this paper are	tabulated in 
\cref{tab:conversion}.

\begin{table}[t] 
\begin{center}\renewcommand*{\arraystretch}{2.6}
\begin{tabular}{| c ||c |c| c| c| } 
	\hline
	& $F$ & $r$ & $p$ & $\chi_{00}$\\
	\hline 
	\hline
	$F$ & $F $& $1-r$& 
	$\dfrac{(d-1)p + 1}{d}$ & $\dfrac{d\chi_{00} + 1}{d+1}$ \\
	\hline
	$r$ & $1-F $ & $r$& $\dfrac{d-1}{d}(1-p)$ & $\dfrac{d}{d+1} 
	(1-\chi_{00} )$\\
	\hline
	\rule{0pt}{3ex}  
	$p$ & $\dfrac{dF - 1}{d-1}$ & $1-\dfrac{d}{d-1}r$& $p$ & 
	$\dfrac{d^2\chi_{00} - 1}{d^2-1}$\\
	\hline
	$\chi_{00}$ & $\dfrac{(d+1)F - 1}{d}$ & $1-\dfrac{d+1}{d}r$& 
	$\dfrac{(d^2-1)p + 1}{d^2}$ & $\chi_{00}$ \\
	\hline
\end{tabular}
\caption{Linear relations between the fidelity $(F)$, the infidelity 
$(r)$, the {effective depolarizing constant} $(p)$, and $\chi_{00}$. }	\label{tab:conversion}
\end{center}
\end{table}

\section{Composite infidelities in terms of component infidelities}

We now prove that unitary error processes lead to the fastest growth in the 
total infidelity of a circuit. In particular, we obtain strict bounds on the 
infidelity of a composite error process in terms of the infidelities of the 
components and show that the bounds are saturated by unitary processes for all 
even-dimensional systems. 

We first obtain a bound on the infidelity of the composition of two channels 
that strictly improves on the corresponding bound of Ref.~\cite{Kimmel2014}. We 
also show that the improved bound is saturated for all values of the relevant 
variables. Therefore \cref{thm:chi00_bound} gives the optimal bounds on the 
infidelity of the composite in terms of only the infidelities of the 
components, and so obtaining a more precise estimate of the composite 
infidelity requires further information about the errors. We then obtain an 
upper bound on the infidelity of the composition of $m$ channels that inherits 
the tightness of the bound for the composition of two channels.

We present the following bounds in terms of the $\chi$ matrix, though the 
results can be rewritten in terms of other linear functions of the infidelity 
using \cref{tab:conversion}. For example, consider the composition of $m$ 
noisy operations $\mc{X}_i$ with equal infidelity, that is, $r(\mc X_i)= r$. 
Then by \cref{thm:chi00_bound_n} and \cref{tab:conversion}, the total 
infidelity of the composite process is at most
	\begin{align}\label{eq:quadratic}
	r( \mc X_{1:m} ) \leq m^2 r+ O(m^4r^{2}),
	\end{align}
which exhibits the expected quadratic scaling with $m$. Moreover, this upper 
bound is saturated and so cannot be improved without additional information 
about the errors.

\begin{thm}\label{thm:chi00_bound}
	For any two quantum channels $\mc{X}$ and $\mc{Y}$, 
	\begin{align}\label{eq:chi00_bound}
	\Big|\chi_{00}^{\mc{X}\mc{Y}} - \chi_{00}^{\mc{X}} \chi_{00}^{\mc{Y}} 
	-(1-\chi_{00}^{\mc{X}})(1-\chi_{00}^{\mc{Y}}) \Big| 
	\notag \\
	\leq 2\sqrt{\chi_{00}^{\mc{X}}\chi_{00}^{\mc{Y}}(1-\chi_{00}^{\mc{X}}) 
		(1-\chi_{00}^{\mc{Y}})}.
	\end{align}
	Furthermore, for all even dimensions and all values of 
	$\chi_{00}^{\mc{X}}$, $\chi_{00}^{\mc{Y}}$, there exists a pair of channels 
	$\mc{X}$ and $\mc{Y}$ saturating both signs of the above inequality.
\end{thm}

\begin{proof}
Let $\mc{X}(\rho) = \sum_j X_j \rho X_j\ct$ and $\mc{Y}(\rho) = \sum_j Y_j 
\rho Y_j\ct$ be Kraus operator decompositions of $\mc{X}$ and $\mc{Y}$ 
respectively. From \cref{eq:Kraus_expansion}, we can expand the Kraus 
operators as 
\begin{align}
	X_j
	&= \sqrt{x_j d}\left(\cos(\xi_j)B_1 + 
	\sin(\xi_j)\vec{u}_j\cdot\vec{\mc{B}}\right) \notag\\
	Y_j
	&= \sqrt{y_j d}\left(\cos(\theta_j)B_1 + 
	\sin(\theta_j)\vec{v}_j\cdot\vec{\mc{B}}\right)
\end{align}
where $\vec{u}_j,\vec{v}_j\in\mbb{C}^{d^2-1}$ are unit vectors and 
$\xi_j,\theta_j\in[0,\tfrac{\pi}{2}]$. Then a Kraus operator decomposition 
of $\mc{X}\mc{Y}$ is 
\begin{align}
	\mc{X}\mc{Y}(\rho) 
	= \sum_{j,k} X_j Y_k \rho Y_k\ct X_j\ct
\end{align}
and so, by \cref{eq:Kraus_to_chi00},
\begin{align}\label{eq:composite_chi00}
	\chi_{00}^{\mc{X}\mc{Y}} 
	= \sum_{j,k} x_j y_k\lv\cos \xi_j \cos \theta_k + 
	\beta_{j,k}\sin\xi_j \sin\theta_k\rv^2,
\end{align}
where $\beta_{j,k} = \vec{u}_j\cdot\vec{v}_k$ and we have chosen the basis 
$\mc{B}$ to be Hermitian so that $\tr B_j\ct B_k =\tr B_j B_k = \delta_{j,k}$. 
By the triangle and reverse-triangle inequalities,
\begin{align}
	|\alpha|-|\gamma|
	\leq|\alpha + \beta\gamma|
	\leq |\alpha| + |\gamma|
\end{align}
for any $\alpha,\beta,\gamma\in\mbb{C}$ such that $|\beta|\leq 1$, which 
then 
implies
\begin{align}\label{eq:triangle_bounds}
	\Bigl\lv |\alpha + \beta\gamma|^2 - |\alpha|^2 - |\gamma|^2 \Bigr\rv
	\leq 2|\alpha\gamma|.
\end{align}
From \cref{eq:Kraus_to_chi00} and \eqref{eq:Kraus_to_rchi00},
\begin{align}
	\sum_{j,k} x_j y_k|\cos(\xi_j)\cos(\theta_k)|^2
	&= \chi_{00}^{\mc{X}}\chi_{00}^{\mc{Y}} \notag\\
	\sum_{j,k} x_j y_k|\sin(\xi_j)\sin(\theta_k)|^2
	&= (1-\chi_{00}^{\mc{X}})(1-\chi_{00}^{\mc{Y}} ),
\end{align}
so by \cref{eq:triangle_bounds},
\begin{align}\label{eq:intermediate}
	\Big|\chi_{00}^{\mc{X}\mc{Y}} - \chi_{00}^{\mc{X}} 
	\chi_{00}^{\mc{Y}} 
	-(1-\chi_{00}^{\mc{X}})(1-\chi_{00}^{\mc{Y}}) \Big| 
	\notag \\
	\leq \sum_{j,k}2x_j 
	y_k\cos(\xi_j)\cos(\theta_k)\sin(\xi_j)\sin(\theta_k),
\end{align}
using $|\beta_{j,k}|\leq 1$ and the non-negativity of the 
trigonometric functions over $[0,\tfrac{\pi}{2}]$. Note that the above 
inequalities are saturated if and only if $\beta_{j,k}=\pm1$.

By the Cauchy-Schwarz inequality with the fact that all the quantities are 
non-negative,
\begin{align*}
\sum_j x_j \sin(\xi_j) \cos(\xi_j)
	&\leq \sqrt{\sum_j x_j \sin^2(\xi_j)}\sqrt{\sum_j x_j \cos^2(\xi_j)} 
	\notag\\
	&\leq \sqrt{(1-\chi_{00}^{\mc{X}})\chi_{00}^{\mc{X}}},
\end{align*}
where the second line follows from \cref{eq:Kraus_to_chi00} and 
\cref{eq:Kraus_to_rchi00}. Applying this upper bound for $\mc{X}$ and the 
corresponding one for $\mc{Y}$ to \cref{eq:intermediate} gives the inequality 
in the theorem.

To see that both signs of the inequality are saturated for all values of 
$\chi_{00}^{\mc{X}},\chi_{00}^{\mc{Y}}$ in even dimensions, let $\mc{X}= 
\mc{U}(\phi)\otimes \mc{I}_{d/2}$ and $\mc{Y} = \mc{U}(\theta)\otimes 
\mc{I}_{d/2}$ where 
\begin{align}
	U(\phi) = e^{i\phi}|0\ra \la 0| + e^{-i\phi}|1\ra \la 1|.
\end{align}
By \cref{eq:Kraus_to_chi00}, $\chi_{00}^{\mc{U}(\phi)\otimes 
\mc{I}_d/2} = 
\chi_{00}^{U(\phi)} = \cos^2\phi$. As $\mc{X}\mc{Y} = 
\mc{U}(\phi+\theta)\otimes\mc{I}_{d/2}$, some trivial trigonometric 
manipulations give
\begin{align}
	&\chi_{00}^{\mc{X}\mc{Y}} - \chi_{00}^{\mc{X}} \chi_{00}^{\mc{Y}} 
	-(1-\chi_{00}^{\mc{X}})(1-\chi_{00}^{\mc{Y}}) \notag \\
	&= - 2\cos\phi \sin\phi \cos\theta \sin\theta \notag\\
	&= -2\sqrt{\chi_{00}^{\mc{X}}\chi_{00}^{\mc{Y}}(1-\chi_{00}^{\mc{X}}) 
		(1-\chi_{00}^{\mc{Y}})} {\rm sign}(\sin 2\phi \sin 2\theta) ,
\end{align}
which saturates the lower bound if the sign function is positive and the 	
upper bound if it is negative.
\end{proof}	

\begin{cor}\label{thm:chi00_bound_n}
For any $m$ quantum channels $\mc{X}_i$ such that
\begin{align}\label{eq:composite_condition}
\sum_{i=1}^m \arccos\sqrt{ \chi_{00}^{\mc{X}_i}} \leq \frac{\pi}{2},
\end{align}
the $\chi_{00}$ element of the composite channel satisfies
\begin{align} \label{eq:composite_bound}
\chi_{00}^{ \mc{X}_{1:m}} \geq \cos^2 \left( \sum_{i=1}^m 
\arccos\sqrt{ \chi_{00}^{\mc{X}_i}}\right).
\end{align}
Furthermore, this bound is saturated for all even dimensions and all values of 
the $\chi_{00}^{\mc{X}_i}$ satisfying \cref{eq:composite_condition}.
\end{cor}

\begin{proof}
We can rewrite the lower bound in \cref{eq:chi00_bound} as
\begin{align}
\sqrt{\chi_{00}^{\mc{X} \mc{Y}}}  \geq 
\sqrt{\chi_{00}^{\mc{X}}}\sqrt{\chi_{00}^{\mc{Y}}}-\sqrt{1-\chi_{00}^{\mc{X}}}
\sqrt{1-\chi_{00}^{\mc{Y}}}.
\end{align}
Writing $\sqrt{\chi_{00}}=\cos(\arccos\sqrt{\chi_{00}})$ and 
$\sqrt{1-\chi_{00}}=\sin(\arccos\sqrt{\chi_{00}})$ and using standard 
trigonometric identities, the above becomes
\begin{align}
\arccos\sqrt{\chi_{00}^{\mc{X} \mc{Y}}}  
\leq  \arccos\sqrt{\chi_{00}^{\mc{X}}} + \arccos\sqrt{\chi_{00}^{\mc{Y}}},
\end{align}
taking note to change the direction of the inequality when taking the 
$\arccos$, which follows from \cref{eq:composite_condition}. By induction, 
we have
\begin{align}
\arccos \left( \sqrt{\chi_{00}^{\mc{X}_{1:m}}} \right) 
&\leq  \sum_i \arccos\left(\sqrt{\chi_{00}^{\mc{X}_i}}\right)
\end{align}
for any set of $m$ channels $\mc{X}_i$. Taking the cosine and squaring gives 
the bound in \cref{eq:composite_bound}. The saturation follows directly from 
the saturation of \cref{eq:chi00_bound}.
\end{proof}	
{A way to intuitively think about \cref{eq:composite_bound} goes as follows: ``the worst possible fidelity of a composition
is obtained through a coherent (unitary) buildup''. Indeed, the trigonometric form of the inequality reflects this coherent nature.}

\section{Improved bounds on the infidelity using the unitarity}\label{sec:unitarity_improvements}

The bounds in \cref{thm:chi00_bound} and \cref{thm:chi00_bound_n} are 
tight for general channels if only the infidelity (or, equivalently, 
$\chi_{00}$) is known. In particular, from \cref{eq:quadratic}, the 
infidelity increases at most quadratically in $m$ (to lowest order in 
$r$). However, the examples that saturate the bounds are all unitary 
channels. If, on the other hand, the error model is a depolarizing channel 
\begin{align}
\mc{D}_p(\rho) = p\rho + \frac{(1-p)}{d}\ds{I}_d,
\end{align}
or a Pauli channel (that is, a channel with a diagonal $\chi$ matrix with 
respect to the Pauli basis), then the infidelity increases at most linearly in 
$m$ to lowest order, that is 
\begin{align}\label{eq:linear}
r(\mc X_{1:m}) \leq mr + O(m^2r^2).
\end{align}
The intermediate regime between Pauli errors and unitary errors can be quantified 
via the unitarity~\cite{Wallman2015}. In particular, we define the (positive) 
coherence angle to be
\begin{align}\label{eq:coherenceangle}
\theta(\mc{E}) = \arccos \left( p(\mc E )/\sqrt{u(\mc E)}\right) .
\end{align}
As $u(\mc{E})\leq 1$ with equality if and only if $\mc{E}$ is unitary,
$\theta (\mc{E}) \in [0,\arccos{p(\mc{E})}]$ and
\begin{align}\label{eq:angle_cases}
\theta(\mc{E}) = \begin{cases}
0 & \mbox{ iff } \mc{E} \mbox{ is depolarizing,} \\
 O(r) & \mbox{ if } \mc{E} \mbox{ is Pauli,} \\
\arccos p(\mc{E}) = O (\sqrt{r}) & \mbox{ iff } \mc{E} \mbox{ is unitary.}
\end{cases}
\end{align}
That is, $\theta(\mc{E})$ quantifies the intermediate regime between Pauli and 
unitary errors for an isolated error process.

We now show that combining the coherence angle and the infidelity enables 
improved bounds on the growth of the infidelity. For example, for any $m$ 
unital channels, or for any $m$ single qubit operations $\mc{X}_i$, with equal 
infidelity $r(\mc X_i)= r$ and coherence 
angles $\theta(\mc{X}_i) = \theta$, the total infidelity is at most
\begin{align}\label{eq:intermediate_regime}
r( \mc X_{1:m}) \leq m\left(r-\frac{(d-1)\theta^2}{2d}\right) +m^2 
\frac{(d-1)\theta^2}{2d}
\end{align}
plus higher-order terms in $r$ and $\theta^2$ by 
\cref{eq:uni_bound_fold_unital}. For Pauli errors, $\theta^2=O(r^2)$, so we 
recover \cref{eq:linear}. Conversely, for unitary errors $(d-1) 
\theta^2=2dr+O(r^2)$, so we recover \cref{eq:quadratic} in such regime. 
Moreover, the above bound is saturated (to the appropriate order) in even 
dimensions by channels of the form 
\begin{align}\label{eq:saturation_cor}
\bc X_i =\left(\begin{array}{cccc}
1 & 0 & 0 & 0 \\
0 & \gamma \cos\theta (\mc X_i)  & -	\gamma\sin\theta(\mc X_i) & 0 \\ 
0 & \gamma\sin\theta(\mc X_i) &	\gamma\cos\theta(\mc X_i) & 0 \\
0 & 0 & 0 & \lambda
\end{array} \right)  \otimes \mbb I_{d^2/2}~.
\end{align}
These include the unital action of single qubit amplitude damping and dephasing channels combined 
with a unitary evolution around the dampening/dephasing axis. {The unitary factor is parameterized by the coherence angle: \mbox{$Z_\theta = \exp{i 2 \theta Z}$} (hence the ``coherence'' qualifier). {In this sense, the coherence angle portays the allowed amount of rotation in Bloch space, as opposed to contractions (quantified by $\gamma,~\lambda$ in our saturation example) due to decoherent effects.}}

\Cref{thm:unitarity_bound,thm:uni_bound_fold} result from more general matrix inequalities that 
we prove in \cref{sec:proofs}. We apply the inequalities to the unital block of 
the Liouville representation from \cref{eq:block}, and substitute the 
expressions for the {effective depolarizing constant} and the unitarity from 
\cref{eq:Liouville_parameters}. For \cref{thm:uni_bound_fold}, we also use 
results from \cite{GarciaPerez2006}, which state that the maximal singular 
value of the unital block is upper-bounded by $\sqrt{\frac{d}{2}}$ for general 
channels and $1$ for unital channels.

\begin{thm}\label{thm:unitarity_bound}
	For any two quantum channels $\mc{X}$ and $\mc{Y}$,
	\begin{align}\label{eq:unitarity_bound}
	\cos[\theta (\mc X) + \theta(\mc Y)]\leq \frac{p(\mc X \mc 
	Y)}{\sqrt{u(\mc X) u(\mc Y)}} &\leq \cos[\theta (\mc X) - \theta(\mc Y)].
	\end{align} 
\end{thm}
{In other words, the leeway in the effective depolarizing constant of a composition $\mc X\mc Y$ is limited by constructive and destructive coherent effects. For longer compositions, we have:}
\begin{thm}\label{thm:uni_bound_fold}
For any $m$ channels $\mc X_i$ with $p(\mc X_i)=p$, $u(\mc X_i)=u$, 
the {effective depolarizing constant} of the composite channel satisfies
\begin{align}\label{eq:uni_bound_fold}
|p(\mc X_{1:m})-p^m|& \leq \sqrt{\frac{d}{2}} {m\choose{2}} u \sin^2(\theta) ~.
\end{align}
Furthermore, if the $\mc{X}_i$ are unital channels, the bound can be improved to
\begin{align}\label{eq:uni_bound_fold_unital}
	|p(\mc X_{1:m})-p^m|& \leq {m\choose{2}} u \sin^2(\theta) ~.
	\end{align}
\end{thm}
{Notice that the binomial factor -- which indicates a quadratic behavior in $m$ -- demonstrates that the effective depolarizing constant of a large composition, $p(\mc X_{1:m})$, can quickly differ from $p^m$. This difference grows quicker with the coherence angle, which can be tied to coherent effects through \cref{eq:saturation_cor}. }

{The bounds in \cref{thm:unitarity_bound} can be made even tighter if
one of the channels is guaranteed to be Pauli. 
	\begin{thm}\label{thm:unitarity_bound_Pauli}
		Consider a Pauli channel $\mc{X}$ and any quantum channel $\mc{Y}$. 
		Then, the composite 
		infidelity is essentially linear in the individual infidelities $r(\mc X)$ and $r(\mc Y)$:
		\begin{align}\label{eq:unitarity_bound_Pauli}
		r(\mc X \mc Y) = r(\mc X) + r(\mc Y) + O(r(\mc X) r(\mc Y)) ~.
		\end{align} 
	\end{thm}
This bound is to be contrasted with the naive usage of
\cref{thm:unitarity_bound}:
		\begin{align}
		r(\mc X \mc Y) & = r(\mc X) + r(\mc Y) + O(\theta(\mc X) \theta{(\mc Y)})  \tag{\cref{thm:unitarity_bound}} \\
		& =  r(\mc X) + r(\mc Y) + O(r(\mc X) \sqrt{r(\mc Y)}) \tag{\cref{eq:angle_cases}}~.
		\end{align} 
The improvement 
can be easily shown as follow. The infidelity is invariant under unitary conjugation
$r(\mc X \mc Y ) =r(\mc U \mc X \mc Y \mc U^\dagger)$ or convex combination of thereof.  In particular,  it is invariant under a Pauli twirl. Since $\mc X$ is a Pauli channel, it commutes with Pauli unitaries, and the twirl gets effectively performed on 
$\mc Y$, which becomes a Pauli channel $\mc Y_{\rm Pauli}$ with 
low coherence angle $\theta(\mc Y_{\rm Pauli})= O(r(\mc Y))$ 
(see \cref{eq:angle_cases}). From there we can apply \cref{thm:unitarity_bound}.}

{\Cref{thm:unitarity_bound,thm:uni_bound_fold} implicitly suggest that using
the coherence angle (rather than the infidelity) as the objective function\footnote{A more stable choice would be $\sin^2(\theta)$.} for optimizing  operations would strongly tighten eventual 
	assertions about the fidelity of circuit constructions.}

\section{Application: Interleaved RB}\label{sec:interleaved_RB}

The fidelity extracted from standard RB experiments typically characterizes the 
average error over a gate set $\mc G$, defined as
\begin{align}
	\mc E := |\mc G|^{-1}\sum_{g \in \mc G} \mc E_g.
\end{align}
However, one might only care about the fidelity $F(\mc{E}_h)$ attached to a 
specific gate of interest $h \in \mc G$, such as one of the generators required 
for universal quantum computing. The interleaved RB protocol (IRB)
\cite{Magesan2012b} yields a fidelity estimate of $\mc E_h \mc E$\footnote{{For the sake of simplicity, we assume that the protocols all provide fidelity estimates defined with respect to the same (or very close) ideal representation of gates.}}, the 
composition between the single gate error and the gate set error, which  
provides bounds on the desired value $F(\mc{E}_h)$. An issue with this approach 
is that these bounds generally have a wide range, since possible coherent 
effects cannot be ignored. This issue is illustrated by the results of two 
simulations of interleaved RB experiments, plotted in \cref{fig:irb_rotvsdep}.
In both scenarios, the fidelity of the gate error and of the composed gate were 
fixed at $F(\mc E)=0.9975$ and $F(\mc E_h \mc E)=0.9960$ respectively, hence 
leading to the same single gate fidelity estimate. In the first case, the 
interleaved gate $h$ is unitary with high fidelity ($F(\mc E_h)=0.9991$), 
whereas in the second case the error is depolarizing, with a lower fidelity 
($F(\mc E_h)=0.9975$). This example illustrates how interleaved RB, without a 
measure of unitarity, can only provide a loose estimate of the infidelity of an 
individual gate.
   
More generally, rearranging the bound in \cref{thm:chi00_bound} to isolate 
$\chi_{00}^{\mc{Y}}$ gives
\begin{align}\label{eq:chi00_bound_comp}
\Big|\chi_{00}^{\mc{Y}} - \chi_{00}^{\mc{X}\mc{Y}} \chi_{00}^{\mc{X}} 
-(1-\chi_{00}^{\mc{X}\mc{Y}})(1-\chi_{00}^{\mc{X}}) \Big| 
\notag \\
\leq 2\sqrt{\chi_{00}^{\mc{X}\mc{Y}} 
	\chi_{00}^{\mc{X}}(1-\chi_{00}^{\mc{X}\mc{Y}}) 	
	(1-\chi_{00}^{\mc{X}})}.
\end{align}
Moreover, this bound cannot be improved without further information. Now 
suppose that $r(\mc{E}_h\mc{E})\approx 2 r(\mc E)$, so that the 	   
uncertainty of $r(\mc E_h)$, obtained via \cref{eq:chi00_bound_comp} and 
\cref{tab:conversion} is
\begin{align}
 \Delta r(\mc E_h) \approx 4 \sqrt{2} r(\mc E).
\end{align}
While this bound does give an estimate of the infidelity, this estimate is 
comparable to the following naive estimate that requires no additional 
experiment. As the fidelity, and hence the infidelity, is a linear function of 
$\mc E$ we have
\begin{align}
 r(\mc E)=|\mc G|^{-1} \sum_{g \in \mc G} r(\mc E_g)
\end{align}
which, since $r(\mc E)$ is non-negative for any channel $\mc E$, implies
\begin{align}
r(\mc E_h) \leq |\mc G| r(\mc E)
\end{align}
for any $h \in \mc G$. (Note also that this bound can be heuristically improved 
by identifying sets of gates that are expected to have comparable error.) When 
$\mc G$ is chosen to be the $12$-element subgroup of the Clifford group that 
forms a unitary $2$-design, the naive bound is, at the very worst, a factor of 
$3/\sqrt{2}$ worse than the bound from interleaved benchmarking and requires no 
additional statistical analysis or data collection.

However, if the error channels were guaranteed to be depolarizing, 
$F(\mc{E}_h)$ could be exactly estimated from an interleaved RB experiment. {In 
	general, we can use our knowledge of the unitarity of $\mc E$ -- which can be obtained from a URB experiment\footnote{{The current analysis of URB is done under a gate-independent noise approximation.}} -- to quantify how close the error model is to 
	depolarizing noise.} From \cref{thm:unitarity_bound}, we then have the following 
bounds, which can be orders of magnitude tighter as illustrated in 
\cref{fig:uni_bound}.
\begin{cor}\label{cor:irb}
	For any two quantum channels $\mc{E}_h$ and $\mc{E}$,
	{\begin{align}\label{eq:uni_bound}
		\left|p(\mc{E}_h) - \frac{p(\mc{E}_h \mc E )p(\mc E )}{u(\mc{E})}\right|
		& \leq \sqrt{1-\frac{p(\mc{E})^2}{u(\mc{E})}}\sqrt{1-\frac{p(\mc E_h\mc{E})^2}{u(\mc{E})}}.
		\end{align}}
\end{cor}	
{Notice that this new estimate of $p(\mc E_h)$ is an amalgam of three experiments: standard RB, IRB and unitarity RB.} {A recommended experimental practice would be, for instance \cite{Yang2018}:
\begin{itemize}
	\item Perform standard RB over the Clifford group.  Estimate the resulting 
	decay parameter $p(\mc E)$, where $\mc E$ is tied to the average 
	error over the Clifford group.
	\item Perform unitarity RB  over the Clifford group. Estimate the resulting decay parameter, 
	which corresponds to the unitarity $u(\mc E)$.
	\item Perform IRB  with the Clifford group as randomizing set and $h$ as ideal gate of interest. Estimate the resulting decay constant $p(\mc E_h \mc E)$, where $\mc E_h$ is the error map attached to $h$.
	\item Use \cref{eq:uni_bound} to bound $p(\mc E_h)$, and use \cref{tab:conversion} to convert it to the fidelity (or infidelity).
\end{itemize}
}{Recall that in the depolarizing case $u(\mc E)= p(\mc E)^2$, 
	for which \cref{eq:uni_bound} reduces to the familiar equality $p(\mc E_h) = p(\mc E_h \mc E)/p(\mc E)$\footnote{In the interleaved RB lingo, this relation is often expressed as $p (\mc E_h) = p_{\rm IRB}/p_{\rm RB}$, where $\mc E_h$ is the error attached to the interleaved gate.}. In fact, the equality remains true up to order $r(\mc E)^2$ in the more general case of stochastic Pauli errors,
	as demonstrated in \cref{thm:unitarity_bound_Pauli}. Treating the
	infidelity as a linear quantity under composition is a very common
	assumption stemming from a classical probabilistic view of error accumulation. To take another example
 	of a linear manipulation, the infidelity
 	per pulse (or infidelity per primitive gate) is often obtained 
 	by implicitly dividing the infidelity of a set of composite gates 
 	by the average number of pulses used to generate them.
 		These are not bad estimates only if the error is mostly stochastic. This might be a valid presumption since many error mechanisms are naturally stochastic, {but is certainly not a trivial one, since coherent effects also commonly arise from faulty control.} The present paper offers a means to {avoid}
 		the {often unrealistic} stochasticity assumption by explicitly providing a confidence interval based on experimental estimates of the unitarity. To illustrate the idea, in \cref{fig:comparison_plot} we applied our bounds on various experimental results \cite{Gaebler2012,Corcoles2013,Kelly2014,Barends2014,Casparis2016,McKay2016,Sheldon2016,Takita2016,McKay2017} and varied {the value of the unitarity}. 
 		}
\begin{figure}
  	\centering
  	\includegraphics[width= 1. \linewidth]{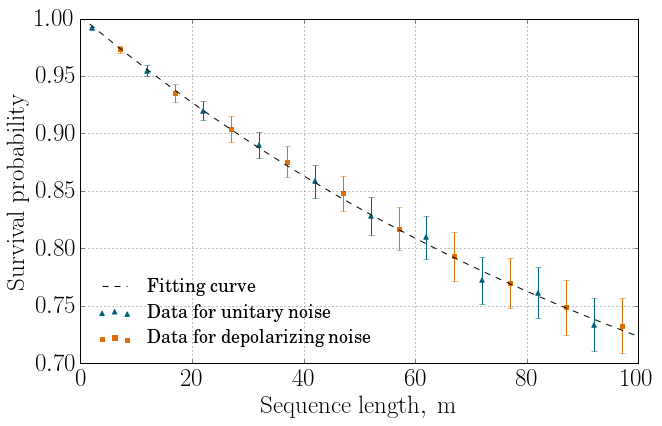}
  	\caption{\label{fig:irb_rotvsdep}
  		(Color online) The average survival probability, $p_{\rm surv.}(m) = |\mc G|^{-m} \sum_i \langle 0 | \mc S^{(i)}_m (|0\rangle \langle 0 | ) | 0\rangle$ 
  		over all sequences $\mc S_m^{(i)}$ of length $m$,  as a function of 
  		the sequence length for two simulated interleaved RB experiments 
  		(see Ref.~\cite{Magesan2012b} for more details) with 
  		two different individual gate errors $\mc E_h$, but a common average 
  		error $\mc E$ with fidelity $0.9975$. Orange squares represent an error
  		model with high fidelity (${F}(\mc E_h)=0.9991$) that interacts 
  		coherently with $\mc E$. Blue triangles represent an error model with 
  		lower fidelity (${F}(\mc E_h)=0.9975$), but that is purely stochastic. 
  		See \cref{sec:interleaved_RB} for more details.}
\end{figure}
          
    \begin{figure*}
    	\centering
    	\includegraphics[width= .75 \linewidth]{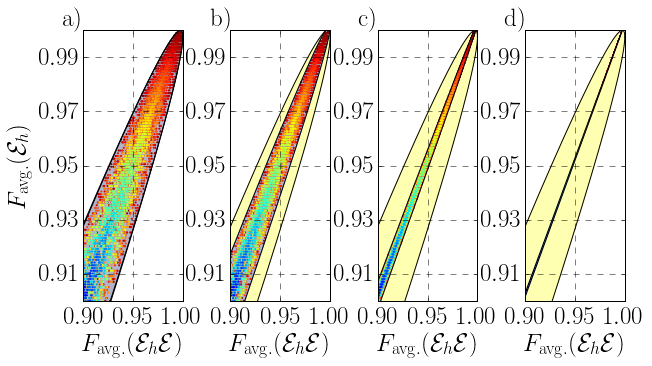}
    	\caption{\label{fig:uni_bound}
   		(Color online) Bounds on the fidelity ${F}(\mc E_h)$ of an individual 
   		gate $h$ as a function of the composite ${F}(\mc E_h \mc E)$ with ${F}(\mc E)$ fixed  
   		and varying $u(\mc E)$:
   		a) $u(\mc E)=1.00000$, b) $u(\mc E)=0.99300$,  
   		c) $u(\mc E)=0.99030$, d) $u(\mc E)=0.99003 \approx p(\mc E)^2$.
   		 The numerical data points 
   		correspond to the values ${F}(\mc E_h)$ and ${F}(\mc E_h \mc E)$ for 
   		randomly-generated channels $\{ \mc E_h, \mc E \}$ satisfying ${F}(\mc 
   		E)=0.9975$ and with the appropriate value of $u(\mc E)$. As illustrated 
   		by the color, the unitarity $u(\mc E_h)$ is minimal in the center of 
   		the shaded region and maximal when the data points approach our bound.}
    \end{figure*}

   \begin{figure*}
   	\centering
   	\includegraphics[width= .75 \linewidth]{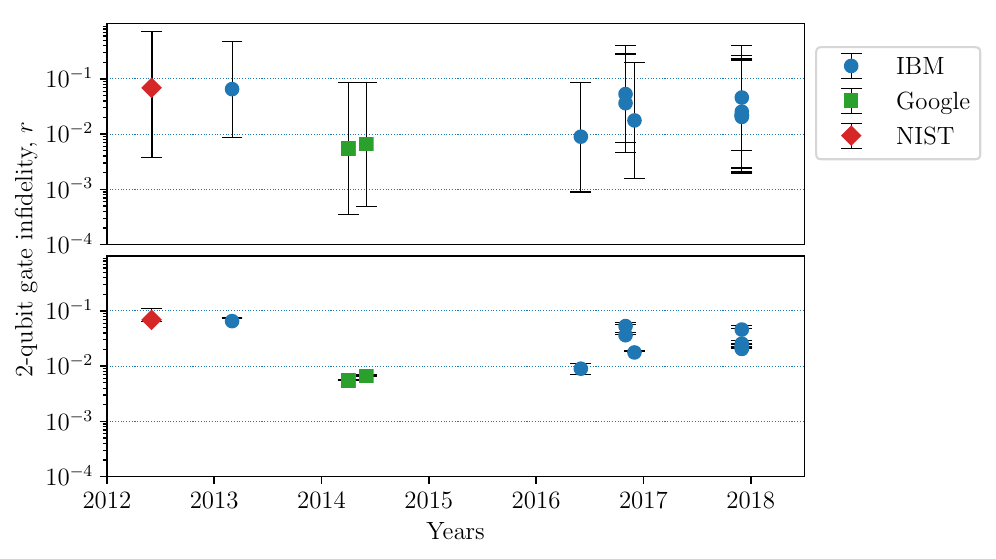}
   	\caption{\label{fig:comparison_plot}
   		(Color online) Bounds on various 2-qubit gate infidelities $r(\mc E_h)$ based on various experimental data \cite{Gaebler2012,Corcoles2013,Kelly2014,Barends2014,Casparis2016,McKay2016,Sheldon2016,Takita2016,McKay2017}. The time refers to the dates of submission. The top 
   		plot uses \cref{eq:uni_bound} with a maximal coherence angle $\theta(\mc E)$, which yields in bounds spanning up to two orders of magnitude. The bottom plot assumes a purely stochastic error model, by which we mean that $u(\mc E) \approx p(\mc E)^2$. For every data point, some statistical error is taken into account, hence the non-zero error bars in the bottom plot.}
   \end{figure*}

\section{Summary and outlook}
{In this paper, we have studied the impact of coherent errors on the 
fidelity of quantum circuits. We first demonstrate why coherent errors are
a serious concern:} a coherent composition of unitary quantum channels 
results in the fastest decay of the fidelity. In this case, the infidelity grows 
quadratically in the number of gates, in contrast with the linear growth for 
stochastic Pauli channels.
{The disparity between these two regimes means that the characterization of the gate fidelities alone only allows to formulate weak statements about the fidelity of more elaborate circuit constructions.}

Hence, in order to characterize {circuits}
more precisely, we introduced a coherence	
angle---{which corresponds to a rotation angle on the Bloch space, as opposed to a contraction} (see \cref{eq:coherenceangle})---which enables a tighter bound on the 
total error in a quantum circuit in terms of robustly estimable quantities 
that smoothly interpolates between the linear and quadratic regimes.

{Our new bound can be used upside-down: from the fidelity of a small circuit construction,
we can bound the fidelity of one of its elements.} As an immediate application, we demonstrated that this bound substantially 
improves the estimates of individual gate fidelities from interleaved 
randomized benchmarking, which, in the absence of the improved bound, are 
comparable to the naive bound obtained by noting that the infidelity from 
standard RB is the average of the infidelities of the individual gates.
{The practicality of \cref{cor:irb} relies on the implicit assumption that the unitarity
obtained from RB as well as the average gate fidelities are 
resulting from closely related gauges\cite{Proctor2017,Wallman2017,Dugas2018}. 
An open problem
would be to relax this assumption by connecting more rigorously the interpretations of different RB experiments. }

\emph{Acknowledgments} --- This research was supported by the U.S. Army 
Research Office through grant W911NF-14-1-0103, CIFAR, the Government of 
Ontario, and the Government of Canada through CFREF, NSERC and Industry Canada.
\clearpage
\onecolumngrid
\appendix

\section{Matrix inequalities on the real field}\label{sec:proofs}

We define the coherence angle of a matrix $M \in \mbb R^{d \times d}$ to be
\begin{align}\label{def:coherent_angle_real}
\theta (M):= \arccos \left(\frac{\tr M}{\sqrt{d}\|M \|_F}\right) \in 
\left[0,\pi\right].
\end{align} 

\begin{thm}\label{thm:bound_2_real}
For any nonzero $M_1$, $M_2$ $\in\mbb R^{d \times d}$ ,
\begin{align}\label{eq:bound_2_real}
	\cos [\theta (M_1) +\theta (M_2)]  \leq \frac{\tr 
	M_1M_2}{\|M_1\|_F\|M_2\|_F} \leq \cos [\theta (M_1)-\theta (M_2)].
\end{align}
Moreover, both bounds are saturated for all values of 
$\|M_1\|_F$, $\|M_2\|_F$, $\theta(M_1)$, and $\theta(M_2)$ in even dimensions.
\end{thm}

\begin{proof}
By the Cauchy-Schwarz inequality,
\begin{align}\label{eq:trprod}
|\tr AB|^2 
= \Bigl(\sum_{i,j} A_{i,j}B_{j,i}\Bigr)^2
\leq \Bigl(\sum_{i,j} A_{i,j}^2\Bigr) \Bigl(\sum_{i,j} B_{i,j}^2\Bigr) 
= (\tr A\ct A)(\tr B\ct B) = \|A\|_F^2 \|B\|_F^2.
\end{align}
Setting $D_i := \frac{\tr M_i}{d} \mbb I_d$ for $i=1,2$,
\begin{align}\label{eq:dif2sin}
\|M_i -D_i \|_F 
&= \sqrt{\tr (M_i\ct M_i - M_i\ct D_i - D_i\ct M_i + D_i\ct D_i)} \notag\\
&= \sqrt{\|M_i\|_F^2 - d^{-1}(\tr M_i)^2} \notag\\
&= \|M_i\|_F\sqrt{1 - \cos^2\theta(M)} \notag\\
&= \|M_i\|_F\sin\theta(M)
\end{align}
using $\tr M\ct = \tr M$, which holds for $M\in \mathbb{R}^{d\times d}$. 
Setting $A=M_1 - D_1$ and $B = M_2 - D_2$ in \cref{eq:trprod} and using 
\cref{eq:dif2sin} gives
\begin{align}\label{eq:sinsin_real}
|\tr (M_1-D_1)(M_2-D_2)|
\leq \|M_1\|_F \|M_2\|_F\sin\theta(M_1) \sin\theta (M_2).
\end{align}
Using \cref{def:coherent_angle_real} on the left-hand side gives
\begin{align}\label{eq:coscos_real}
|\tr (M_1-D_1)(M_2-D_2)|
&= |\tr M_1 M_2-d^{-1}\tr M_1 \tr M_2|  \notag\\
&= \left|\tr M_1 M_2- \|M_1\|_F 
\|M_2\|_F \cos[\theta (M_1)]\cos[\theta (M_2)] \right|.
\end{align}
Combining \cref{eq:sinsin_real,eq:coscos_real} with the identity $\cos(a \pm 
b)=\cos(a)\cos(b) \mp \sin(a) \sin(b)$ gives both desired inequalities. 
For even $d$, the bounds of \cref{eq:bound_2_real} are saturated by
\begin{align}\label{eq:saturation_real}
\frac{\|M_i\|_F}{\sqrt{d}} \left(\begin{array}{cc}
\cos\theta (M_i)  &- \sin\theta(M_i) \\ 
 \sin\theta(M_i) &\cos\theta(M_i)
\end{array} \right)  \otimes \mbb  I_{\frac{d}{2}}.
\end{align}
\end{proof}

We can generalize the lower bound of \cref{eq:bound_2_real} to matrix products 
$M_{1:m} := M_1 M_2\cdots M_m$.

\begin{thm}
Let $M_1,\ldots,M_m \in\mbb R^{d \times d}$ be such that for all $j$, 
$\theta(M_j)=\theta$, $\frac{\tr(M_j)}{d}=p \leq 1$, $\frac{\|M_j\|_F^2}{d}=u 
\leq 1$, and $\|M_{1:j}\|_2 \leq \sigma_{\rm max}$. Then,
\begin{align}\label{eq:bound_m_real}
\left|\frac{\tr M_{1:m} }{d}-p^m \right|  \leq \sigma_{\rm max}\left({\frac{1-mp^{m-1}-(m-1)p^m}{(1-p)^2}}\right) u \sin^2(\theta) \leq  \sigma_{\rm max} {m \choose 2} u \sin^2(\theta).
\end{align}
\end{thm}

\begin{proof}
Let $D:=p \mbb I_d$, and $M_j= D+ \Delta_j$. Using a telescoping expansion 
twice gives
\begin{align}
M_{1:m}-D^m &=\sum_{i=1}^m M_{1:i-1} (M_i-D) D^{m-i}  \notag \\
& =\sum_{i=1}^m [D^{i-1}+\sum_{j=1}^{i-1}M_{1:j-1} \Delta_j D^{i-1-j}] \Delta_i 
D^{m-i}~.
\end{align}
Taking the trace of each side and using $\tr \Delta_j=0$ gives
\begin{align}
\tr{M_{1:m}} - d p^m = \sum_{i=1}^m \sum_{j=1}^{i-1} p^{m-j-1} \tr{M_{1:j-1} 
\Delta_j \Delta_i } ~.
\end{align}

Therefore
\begin{align}
\left|\tr M_{1:m} -dp^m \right|   & = \left| \sum_{i=1}^m \sum_{j=1}^{i-1} 
p^{m-j-1} \tr{M_{1:j-1} \Delta_j \Delta_i } \right| \notag \\
& \leq \sum_{i=1}^m \sum_{j=1}^{i-1} p^{m-j-1}  |\tr{M_{1:j-1} \Delta_j 
\Delta_i }|  \tag{$\triangle$ inequality} \\
& \leq \sum_{i=1}^m \sum_{j=1}^{i-1} p^{m-j-1}  \|M_{1:j-1} \Delta_j \|_F 
\|\Delta_i\|_F \tag{Cauchy-Schwarz inequality} \\
& \leq \sum_{i=1}^m \sum_{j=1}^{i-1} p^{m-j-1} \sigma_{\rm max} \| \Delta_j 
\|_F \|\Delta_i\|_F  \tag{\cite[Prop. 9.3.6]{Bhatia1997}} \\
& = \sigma_{\rm max} d u \sin^2(\theta) \sum_{i=1}^m \sum_{j=1}^{i-1} 
p^{m-j-1}~,
\end{align}
where we used $\|\Delta_j \|_F= \sqrt{du} \sin(\theta)$ on the last line. 
Let $S:=\sum_{i=1}^m \sum_{j=1}^{i-1} p^{m-j-1}=\sum_{i=1}^{m-1} i p^{i-1}$. 
Using a telescoping expansion leads to 
\begin{align}
S-pS & = -(m-1)p^{m-1}+\sum_{i=0}^{m-2} p^i \notag \\
& =\frac{1-p^{m-1}}{1-p}-(m-1)p^{m-1} \notag \\
\Rightarrow S & =\frac{1-mp^{m-1}-(m-1)p^m}{(1-p)^2}~.
\end{align}
$S$ is maximized when $p=1$, in which case it equals ${m \choose 2}$.
\end{proof}

\bibliography{qcvv}

\begin{thebibliography}{49}%
\makeatletter
\providecommand \@ifxundefined [1]{%
 \@ifx{#1\undefined}
}%
\providecommand \@ifnum [1]{%
 \ifnum #1\expandafter \@firstoftwo
 \else \expandafter \@secondoftwo
 \fi
}%
\providecommand \@ifx [1]{%
 \ifx #1\expandafter \@firstoftwo
 \else \expandafter \@secondoftwo
 \fi
}%
\providecommand \natexlab [1]{#1}%
\providecommand \enquote  [1]{``#1''}%
\providecommand \bibnamefont  [1]{#1}%
\providecommand \bibfnamefont [1]{#1}%
\providecommand \citenamefont [1]{#1}%
\providecommand \href@noop [0]{\@secondoftwo}%
\providecommand \href [0]{\begingroup \@sanitize@url \@href}%
\providecommand \@href[1]{\@@startlink{#1}\@@href}%
\providecommand \@@href[1]{\endgroup#1\@@endlink}%
\providecommand \@sanitize@url [0]{\catcode `\\12\catcode `\$12\catcode
  `\&12\catcode `\#12\catcode `\^12\catcode `\_12\catcode `\%12\relax}%
\providecommand \@@startlink[1]{}%
\providecommand \@@endlink[0]{}%
\providecommand \url  [0]{\begingroup\@sanitize@url \@url }%
\providecommand \@url [1]{\endgroup\@href {#1}{\urlprefix }}%
\providecommand \urlprefix  [0]{URL }%
\providecommand \Eprint [0]{\href }%
\providecommand \doibase [0]{http://dx.doi.org/}%
\providecommand \selectlanguage [0]{\@gobble}%
\providecommand \bibinfo  [0]{\@secondoftwo}%
\providecommand \bibfield  [0]{\@secondoftwo}%
\providecommand \translation [1]{[#1]}%
\providecommand \BibitemOpen [0]{}%
\providecommand \bibitemStop [0]{}%
\providecommand \bibitemNoStop [0]{.\EOS\space}%
\providecommand \EOS [0]{\spacefactor3000\relax}%
\providecommand \BibitemShut  [1]{\csname bibitem#1\endcsname}%
\let\auto@bib@innerbib\@empty
\bibitem [{\citenamefont {Gaebler}\ \emph {et~al.}(2012)\citenamefont
  {Gaebler}, \citenamefont {Meier}, \citenamefont {Tan}, \citenamefont
  {Bowler}, \citenamefont {Lin}, \citenamefont {Hanneke}, \citenamefont {Jost},
  \citenamefont {Home}, \citenamefont {Knill}, \citenamefont {Leibfried},\ and\
  \citenamefont {Wineland}}]{Gaebler2012}%
  \BibitemOpen
  \bibfield  {author} {\bibinfo {author} {\bibfnamefont {J.~P.}\ \bibnamefont
  {Gaebler}}, \bibinfo {author} {\bibfnamefont {A.~M.}\ \bibnamefont {Meier}},
  \bibinfo {author} {\bibfnamefont {T.~R.}\ \bibnamefont {Tan}}, \bibinfo
  {author} {\bibfnamefont {R.}~\bibnamefont {Bowler}}, \bibinfo {author}
  {\bibfnamefont {Y.}~\bibnamefont {Lin}}, \bibinfo {author} {\bibfnamefont
  {D.}~\bibnamefont {Hanneke}}, \bibinfo {author} {\bibfnamefont {J.~D.}\
  \bibnamefont {Jost}}, \bibinfo {author} {\bibfnamefont {J.~P.}\ \bibnamefont
  {Home}}, \bibinfo {author} {\bibfnamefont {E.}~\bibnamefont {Knill}},
  \bibinfo {author} {\bibfnamefont {D.}~\bibnamefont {Leibfried}}, \ and\
  \bibinfo {author} {\bibfnamefont {D.~J.}\ \bibnamefont {Wineland}},\ }\href
  {\doibase 10.1103/PhysRevLett.108.260503} {\bibfield  {journal} {\bibinfo
  {journal} {Physical Review Letters}\ }\textbf {\bibinfo {volume} {108}},\
  \bibinfo {pages} {260503} (\bibinfo {year} {2012})}\BibitemShut {NoStop}%
\bibitem [{\citenamefont {C{\'{o}}rcoles}\ \emph {et~al.}(2013)\citenamefont
  {C{\'{o}}rcoles}, \citenamefont {Gambetta}, \citenamefont {Chow},
  \citenamefont {Smolin}, \citenamefont {Ware}, \citenamefont {Strand},
  \citenamefont {Plourde},\ and\ \citenamefont {Steffen}}]{Corcoles2013}%
  \BibitemOpen
  \bibfield  {author} {\bibinfo {author} {\bibfnamefont {A.~D.}\ \bibnamefont
  {C{\'{o}}rcoles}}, \bibinfo {author} {\bibfnamefont {J.~M.}\ \bibnamefont
  {Gambetta}}, \bibinfo {author} {\bibfnamefont {J.~M.}\ \bibnamefont {Chow}},
  \bibinfo {author} {\bibfnamefont {J.~A.}\ \bibnamefont {Smolin}}, \bibinfo
  {author} {\bibfnamefont {M.}~\bibnamefont {Ware}}, \bibinfo {author}
  {\bibfnamefont {J.}~\bibnamefont {Strand}}, \bibinfo {author} {\bibfnamefont
  {B.~L.~T.}\ \bibnamefont {Plourde}}, \ and\ \bibinfo {author} {\bibfnamefont
  {M.}~\bibnamefont {Steffen}},\ }\href {\doibase 10.1103/PhysRevA.87.030301}
  {\bibfield  {journal} {\bibinfo  {journal} {Physical Review A}\ }\textbf
  {\bibinfo {volume} {87}},\ \bibinfo {pages} {030301} (\bibinfo {year}
  {2013})}\BibitemShut {NoStop}%
\bibitem [{\citenamefont {{Kelly}}\ \emph {et~al.}(2014)\citenamefont
  {{Kelly}}, \citenamefont {{Barends}}, \citenamefont {{Campbell}},
  \citenamefont {{Chen}}, \citenamefont {{Chen}}, \citenamefont {{Chiaro}},
  \citenamefont {{Dunsworth}}, \citenamefont {{Fowler}}, \citenamefont {{Hoi}},
  \citenamefont {{Jeffrey}}, \citenamefont {{Megrant}}, \citenamefont
  {{Mutus}}, \citenamefont {{Neill}}, \citenamefont {{O'Malley}}, \citenamefont
  {{Quintana}}, \citenamefont {{Roushan}}, \citenamefont {{Sank}},
  \citenamefont {{Vainsencher}}, \citenamefont {{Wenner}}, \citenamefont
  {{White}}, \citenamefont {{Cleland}},\ and\ \citenamefont
  {{Martinis}}}]{Kelly2014}%
  \BibitemOpen
  \bibfield  {author} {\bibinfo {author} {\bibfnamefont {J.}~\bibnamefont
  {{Kelly}}}, \bibinfo {author} {\bibfnamefont {R.}~\bibnamefont {{Barends}}},
  \bibinfo {author} {\bibfnamefont {B.}~\bibnamefont {{Campbell}}}, \bibinfo
  {author} {\bibfnamefont {Y.}~\bibnamefont {{Chen}}}, \bibinfo {author}
  {\bibfnamefont {Z.}~\bibnamefont {{Chen}}}, \bibinfo {author} {\bibfnamefont
  {B.}~\bibnamefont {{Chiaro}}}, \bibinfo {author} {\bibfnamefont
  {A.}~\bibnamefont {{Dunsworth}}}, \bibinfo {author} {\bibfnamefont {A.~G.}\
  \bibnamefont {{Fowler}}}, \bibinfo {author} {\bibfnamefont {I.-C.}\
  \bibnamefont {{Hoi}}}, \bibinfo {author} {\bibfnamefont {E.}~\bibnamefont
  {{Jeffrey}}}, \bibinfo {author} {\bibfnamefont {A.}~\bibnamefont
  {{Megrant}}}, \bibinfo {author} {\bibfnamefont {J.}~\bibnamefont {{Mutus}}},
  \bibinfo {author} {\bibfnamefont {C.}~\bibnamefont {{Neill}}}, \bibinfo
  {author} {\bibfnamefont {P.~J.~J.}\ \bibnamefont {{O'Malley}}}, \bibinfo
  {author} {\bibfnamefont {C.}~\bibnamefont {{Quintana}}}, \bibinfo {author}
  {\bibfnamefont {P.}~\bibnamefont {{Roushan}}}, \bibinfo {author}
  {\bibfnamefont {D.}~\bibnamefont {{Sank}}}, \bibinfo {author} {\bibfnamefont
  {A.}~\bibnamefont {{Vainsencher}}}, \bibinfo {author} {\bibfnamefont
  {J.}~\bibnamefont {{Wenner}}}, \bibinfo {author} {\bibfnamefont {T.~C.}\
  \bibnamefont {{White}}}, \bibinfo {author} {\bibfnamefont {A.~N.}\
  \bibnamefont {{Cleland}}}, \ and\ \bibinfo {author} {\bibfnamefont {J.~M.}\
  \bibnamefont {{Martinis}}},\ }\href {\doibase 10.1103/PhysRevLett.112.240504}
  {\bibfield  {journal} {\bibinfo  {journal} {Physical Review Letters}\
  }\textbf {\bibinfo {volume} {112}},\ \bibinfo {eid} {240504} (\bibinfo {year}
  {2014})},\ \Eprint {http://arxiv.org/abs/1403.0035} {arXiv:1403.0035
  [quant-ph]} \BibitemShut {NoStop}%
\bibitem [{\citenamefont {Barends}\ \emph {et~al.}(2014)\citenamefont
  {Barends}, \citenamefont {Kelly}, \citenamefont {Megrant}, \citenamefont
  {Veitia}, \citenamefont {Sank}, \citenamefont {Jeffrey}, \citenamefont
  {White}, \citenamefont {Mutus}, \citenamefont {Fowler}, \citenamefont
  {Campbell}, \citenamefont {Chen}, \citenamefont {Chen}, \citenamefont
  {Chiaro}, \citenamefont {Dunsworth}, \citenamefont {Neill}, \citenamefont
  {O/'Malley}, \citenamefont {Roushan}, \citenamefont {Vainsencher},
  \citenamefont {Wenner}, \citenamefont {Korotkov}, \citenamefont {Cleland},\
  and\ \citenamefont {Martinis}}]{Barends2014}%
  \BibitemOpen
  \bibfield  {author} {\bibinfo {author} {\bibfnamefont {R.}~\bibnamefont
  {Barends}}, \bibinfo {author} {\bibfnamefont {J.}~\bibnamefont {Kelly}},
  \bibinfo {author} {\bibfnamefont {A.}~\bibnamefont {Megrant}}, \bibinfo
  {author} {\bibfnamefont {A.}~\bibnamefont {Veitia}}, \bibinfo {author}
  {\bibfnamefont {D.}~\bibnamefont {Sank}}, \bibinfo {author} {\bibfnamefont
  {E.}~\bibnamefont {Jeffrey}}, \bibinfo {author} {\bibfnamefont {T.~C.}\
  \bibnamefont {White}}, \bibinfo {author} {\bibfnamefont {J.}~\bibnamefont
  {Mutus}}, \bibinfo {author} {\bibfnamefont {A.~G.}\ \bibnamefont {Fowler}},
  \bibinfo {author} {\bibfnamefont {B.}~\bibnamefont {Campbell}}, \bibinfo
  {author} {\bibfnamefont {Y.}~\bibnamefont {Chen}}, \bibinfo {author}
  {\bibfnamefont {Z.}~\bibnamefont {Chen}}, \bibinfo {author} {\bibfnamefont
  {B.}~\bibnamefont {Chiaro}}, \bibinfo {author} {\bibfnamefont
  {A.}~\bibnamefont {Dunsworth}}, \bibinfo {author} {\bibfnamefont
  {C.}~\bibnamefont {Neill}}, \bibinfo {author} {\bibfnamefont
  {P.}~\bibnamefont {O/'Malley}}, \bibinfo {author} {\bibfnamefont
  {P.}~\bibnamefont {Roushan}}, \bibinfo {author} {\bibfnamefont
  {A.}~\bibnamefont {Vainsencher}}, \bibinfo {author} {\bibfnamefont
  {J.}~\bibnamefont {Wenner}}, \bibinfo {author} {\bibfnamefont {A.~N.}\
  \bibnamefont {Korotkov}}, \bibinfo {author} {\bibfnamefont {A.~N.}\
  \bibnamefont {Cleland}}, \ and\ \bibinfo {author} {\bibfnamefont {J.~M.}\
  \bibnamefont {Martinis}},\ }\href {http://dx.doi.org/10.1038/nature13171}
  {\bibfield  {journal} {\bibinfo  {journal} {Nature}\ }\textbf {\bibinfo
  {volume} {508}},\ \bibinfo {pages} {500} (\bibinfo {year} {2014})},\ \bibinfo
  {note} {letter}\BibitemShut {NoStop}%
\bibitem [{\citenamefont {{Xia}}\ \emph {et~al.}(2015)\citenamefont {{Xia}},
  \citenamefont {{Lichtman}}, \citenamefont {{Maller}}, \citenamefont {{Carr}},
  \citenamefont {{Piotrowicz}}, \citenamefont {{Isenhower}},\ and\
  \citenamefont {{Saffman}}}]{Xia2015}%
  \BibitemOpen
  \bibfield  {author} {\bibinfo {author} {\bibfnamefont {T.}~\bibnamefont
  {{Xia}}}, \bibinfo {author} {\bibfnamefont {M.}~\bibnamefont {{Lichtman}}},
  \bibinfo {author} {\bibfnamefont {K.}~\bibnamefont {{Maller}}}, \bibinfo
  {author} {\bibfnamefont {A.~W.}\ \bibnamefont {{Carr}}}, \bibinfo {author}
  {\bibfnamefont {M.~J.}\ \bibnamefont {{Piotrowicz}}}, \bibinfo {author}
  {\bibfnamefont {L.}~\bibnamefont {{Isenhower}}}, \ and\ \bibinfo {author}
  {\bibfnamefont {M.}~\bibnamefont {{Saffman}}},\ }\href {\doibase
  10.1103/PhysRevLett.114.100503} {\bibfield  {journal} {\bibinfo  {journal}
  {\prl}\ }\textbf {\bibinfo {volume} {114}},\ \bibinfo {eid} {100503}
  (\bibinfo {year} {2015})},\ \Eprint {http://arxiv.org/abs/1501.02041}
  {arXiv:1501.02041 [quant-ph]} \BibitemShut {NoStop}%
\bibitem [{\citenamefont {Muhonen}\ \emph {et~al.}(2015)\citenamefont
  {Muhonen}, \citenamefont {Laucht}, \citenamefont {Simmons}, \citenamefont
  {Dehollain}, \citenamefont {Kalra}, \citenamefont {Hudson}, \citenamefont
  {Freer}, \citenamefont {Itoh}, \citenamefont {Jamieson}, \citenamefont
  {McCallum}, \citenamefont {Dzurak},\ and\ \citenamefont
  {Morello}}]{Muhonen2015}%
  \BibitemOpen
  \bibfield  {author} {\bibinfo {author} {\bibfnamefont {J.~T.}\ \bibnamefont
  {Muhonen}}, \bibinfo {author} {\bibfnamefont {A.}~\bibnamefont {Laucht}},
  \bibinfo {author} {\bibfnamefont {S.}~\bibnamefont {Simmons}}, \bibinfo
  {author} {\bibfnamefont {J.~P.}\ \bibnamefont {Dehollain}}, \bibinfo {author}
  {\bibfnamefont {R.}~\bibnamefont {Kalra}}, \bibinfo {author} {\bibfnamefont
  {F.~E.}\ \bibnamefont {Hudson}}, \bibinfo {author} {\bibfnamefont
  {S.}~\bibnamefont {Freer}}, \bibinfo {author} {\bibfnamefont {K.~M.}\
  \bibnamefont {Itoh}}, \bibinfo {author} {\bibfnamefont {D.~N.}\ \bibnamefont
  {Jamieson}}, \bibinfo {author} {\bibfnamefont {J.~C.}\ \bibnamefont
  {McCallum}}, \bibinfo {author} {\bibfnamefont {A.~S.}\ \bibnamefont
  {Dzurak}}, \ and\ \bibinfo {author} {\bibfnamefont {A.}~\bibnamefont
  {Morello}},\ }\href {http://stacks.iop.org/0953-8984/27/i=15/a=154205}
  {\bibfield  {journal} {\bibinfo  {journal} {Journal of Physics: Condensed
  Matter}\ }\textbf {\bibinfo {volume} {27}},\ \bibinfo {pages} {154205}
  (\bibinfo {year} {2015})}\BibitemShut {NoStop}%
\bibitem [{\citenamefont {Tarlton}()}]{Tarlton}%
  \BibitemOpen
  \bibfield  {author} {\bibinfo {author} {\bibfnamefont {J.~E.}\ \bibnamefont
  {Tarlton}},\ }\emph {\bibinfo {title} {Probing qubit memory errors at the 10-
  5 level}},\ \href@noop {} {Ph.D. thesis}\BibitemShut {NoStop}%
\bibitem [{\citenamefont {{Casparis}}\ \emph {et~al.}(2016)\citenamefont
  {{Casparis}}, \citenamefont {{Larsen}}, \citenamefont {{Olsen}},
  \citenamefont {{Kuemmeth}}, \citenamefont {{Krogstrup}}, \citenamefont
  {{Nyg{\^a}rd}}, \citenamefont {{Petersson}},\ and\ \citenamefont
  {{Marcus}}}]{Casparis2016}%
  \BibitemOpen
  \bibfield  {author} {\bibinfo {author} {\bibfnamefont {L.}~\bibnamefont
  {{Casparis}}}, \bibinfo {author} {\bibfnamefont {T.~W.}\ \bibnamefont
  {{Larsen}}}, \bibinfo {author} {\bibfnamefont {M.~S.}\ \bibnamefont
  {{Olsen}}}, \bibinfo {author} {\bibfnamefont {F.}~\bibnamefont {{Kuemmeth}}},
  \bibinfo {author} {\bibfnamefont {P.}~\bibnamefont {{Krogstrup}}}, \bibinfo
  {author} {\bibfnamefont {J.}~\bibnamefont {{Nyg{\^a}rd}}}, \bibinfo {author}
  {\bibfnamefont {K.~D.}\ \bibnamefont {{Petersson}}}, \ and\ \bibinfo {author}
  {\bibfnamefont {C.~M.}\ \bibnamefont {{Marcus}}},\ }\href {\doibase
  10.1103/PhysRevLett.116.150505} {\bibfield  {journal} {\bibinfo  {journal}
  {\prl}\ }\textbf {\bibinfo {volume} {116}},\ \bibinfo {eid} {150505}
  (\bibinfo {year} {2016})}\BibitemShut {NoStop}%
\bibitem [{\citenamefont {{McKay}}\ \emph {et~al.}(2016)\citenamefont
  {{McKay}}, \citenamefont {{Filipp}}, \citenamefont {{Mezzacapo}},
  \citenamefont {{Magesan}}, \citenamefont {{Chow}},\ and\ \citenamefont
  {{Gambetta}}}]{McKay2016}%
  \BibitemOpen
  \bibfield  {author} {\bibinfo {author} {\bibfnamefont {D.~C.}\ \bibnamefont
  {{McKay}}}, \bibinfo {author} {\bibfnamefont {S.}~\bibnamefont {{Filipp}}},
  \bibinfo {author} {\bibfnamefont {A.}~\bibnamefont {{Mezzacapo}}}, \bibinfo
  {author} {\bibfnamefont {E.}~\bibnamefont {{Magesan}}}, \bibinfo {author}
  {\bibfnamefont {J.~M.}\ \bibnamefont {{Chow}}}, \ and\ \bibinfo {author}
  {\bibfnamefont {J.~M.}\ \bibnamefont {{Gambetta}}},\ }\href {\doibase
  10.1103/PhysRevApplied.6.064007} {\bibfield  {journal} {\bibinfo  {journal}
  {Physical Review Applied}\ }\textbf {\bibinfo {volume} {6}},\ \bibinfo {eid}
  {064007} (\bibinfo {year} {2016})}\BibitemShut {NoStop}%
\bibitem [{\citenamefont {{Sheldon}}\ \emph {et~al.}(2016)\citenamefont
  {{Sheldon}}, \citenamefont {{Magesan}}, \citenamefont {{Chow}},\ and\
  \citenamefont {{Gambetta}}}]{Sheldon2016}%
  \BibitemOpen
  \bibfield  {author} {\bibinfo {author} {\bibfnamefont {S.}~\bibnamefont
  {{Sheldon}}}, \bibinfo {author} {\bibfnamefont {E.}~\bibnamefont
  {{Magesan}}}, \bibinfo {author} {\bibfnamefont {J.~M.}\ \bibnamefont
  {{Chow}}}, \ and\ \bibinfo {author} {\bibfnamefont {J.~M.}\ \bibnamefont
  {{Gambetta}}},\ }\href {\doibase 10.1103/PhysRevA.93.060302} {\bibfield
  {journal} {\bibinfo  {journal} {Physical Review A}\ }\textbf {\bibinfo
  {volume} {93}},\ \bibinfo {eid} {060302} (\bibinfo {year}
  {2016})}\BibitemShut {NoStop}%
\bibitem [{\citenamefont {{Takita}}\ \emph {et~al.}(2016)\citenamefont
  {{Takita}}, \citenamefont {{C{\'o}rcoles}}, \citenamefont {{Magesan}},
  \citenamefont {{Abdo}}, \citenamefont {{Brink}}, \citenamefont {{Cross}},
  \citenamefont {{Chow}},\ and\ \citenamefont {{Gambetta}}}]{Takita2016}%
  \BibitemOpen
  \bibfield  {author} {\bibinfo {author} {\bibfnamefont {M.}~\bibnamefont
  {{Takita}}}, \bibinfo {author} {\bibfnamefont {A.~D.}\ \bibnamefont
  {{C{\'o}rcoles}}}, \bibinfo {author} {\bibfnamefont {E.}~\bibnamefont
  {{Magesan}}}, \bibinfo {author} {\bibfnamefont {B.}~\bibnamefont {{Abdo}}},
  \bibinfo {author} {\bibfnamefont {M.}~\bibnamefont {{Brink}}}, \bibinfo
  {author} {\bibfnamefont {A.}~\bibnamefont {{Cross}}}, \bibinfo {author}
  {\bibfnamefont {J.~M.}\ \bibnamefont {{Chow}}}, \ and\ \bibinfo {author}
  {\bibfnamefont {J.~M.}\ \bibnamefont {{Gambetta}}},\ }\href {\doibase
  10.1103/PhysRevLett.117.210505} {\bibfield  {journal} {\bibinfo  {journal}
  {\prl}\ }\textbf {\bibinfo {volume} {117}},\ \bibinfo {eid} {210505}
  (\bibinfo {year} {2016})}\BibitemShut {NoStop}%
\bibitem [{\citenamefont {{McKay}}\ \emph
  {et~al.}(2017{\natexlab{a}})\citenamefont {{McKay}}, \citenamefont
  {{Sheldon}}, \citenamefont {{Smolin}}, \citenamefont {{Chow}},\ and\
  \citenamefont {{Gambetta}}}]{McKay2017}%
  \BibitemOpen
  \bibfield  {author} {\bibinfo {author} {\bibfnamefont {D.~C.}\ \bibnamefont
  {{McKay}}}, \bibinfo {author} {\bibfnamefont {S.}~\bibnamefont {{Sheldon}}},
  \bibinfo {author} {\bibfnamefont {J.~A.}\ \bibnamefont {{Smolin}}}, \bibinfo
  {author} {\bibfnamefont {J.~M.}\ \bibnamefont {{Chow}}}, \ and\ \bibinfo
  {author} {\bibfnamefont {J.~M.}\ \bibnamefont {{Gambetta}}},\ }\href@noop {}
  {\bibfield  {journal} {\bibinfo  {journal} {ArXiv e-prints}\ ,\ \bibinfo
  {eid} {arXiv:1712.06550}} (\bibinfo {year} {2017}{\natexlab{a}})},\ \Eprint
  {http://arxiv.org/abs/1712.06550} {arXiv:1712.06550 [quant-ph]} \BibitemShut
  {NoStop}%
\bibitem [{\citenamefont {Emerson}\ \emph {et~al.}(2005)\citenamefont
  {Emerson}, \citenamefont {Alicki},\ and\ \citenamefont
  {\.{Z}yczkowski}}]{Emerson2005}%
  \BibitemOpen
  \bibfield  {author} {\bibinfo {author} {\bibfnamefont {J.}~\bibnamefont
  {Emerson}}, \bibinfo {author} {\bibfnamefont {R.}~\bibnamefont {Alicki}}, \
  and\ \bibinfo {author} {\bibfnamefont {K.}~\bibnamefont {\.{Z}yczkowski}},\
  }\href {\doibase 10.1088/1464-4266/7/10/021} {\bibfield  {journal} {\bibinfo
  {journal} {Journal of Optics B: Quantum and Semiclassical Optics}\ }\textbf
  {\bibinfo {volume} {7}},\ \bibinfo {pages} {S347} (\bibinfo {year}
  {2005})}\BibitemShut {NoStop}%
\bibitem [{\citenamefont {Dankert}\ \emph {et~al.}(2006)\citenamefont
  {Dankert}, \citenamefont {Cleve}, \citenamefont {Emerson},\ and\
  \citenamefont {Livine}}]{DCEL2006}%
  \BibitemOpen
  \bibfield  {author} {\bibinfo {author} {\bibfnamefont {C.}~\bibnamefont
  {Dankert}}, \bibinfo {author} {\bibfnamefont {R.}~\bibnamefont {Cleve}},
  \bibinfo {author} {\bibfnamefont {J.}~\bibnamefont {Emerson}}, \ and\
  \bibinfo {author} {\bibfnamefont {E.}~\bibnamefont {Livine}},\ }\href@noop {}
  {\bibfield  {journal} {\bibinfo  {journal} {Phys. Rev. A}\ }\textbf {\bibinfo
  {volume} {80}} (\bibinfo {year} {2009); arxiv.org/abs/quant-ph/0606161
  (2006})}\BibitemShut {NoStop}%
\bibitem [{\citenamefont {Magesan}\ \emph {et~al.}(2011)\citenamefont
  {Magesan}, \citenamefont {Gambetta},\ and\ \citenamefont
  {Emerson}}]{Magesan2011}%
  \BibitemOpen
  \bibfield  {author} {\bibinfo {author} {\bibfnamefont {E.}~\bibnamefont
  {Magesan}}, \bibinfo {author} {\bibfnamefont {J.~M.}\ \bibnamefont
  {Gambetta}}, \ and\ \bibinfo {author} {\bibfnamefont {J.}~\bibnamefont
  {Emerson}},\ }\href {\doibase 10.1103/PhysRevLett.106.180504} {\bibfield
  {journal} {\bibinfo  {journal} {Physical Review Letters}\ }\textbf {\bibinfo
  {volume} {106}},\ \bibinfo {pages} {180504} (\bibinfo {year}
  {2011})}\BibitemShut {NoStop}%
\bibitem [{\citenamefont {Magesan}\ \emph {et~al.}(2012)\citenamefont
  {Magesan}, \citenamefont {Gambetta}, \citenamefont {Johnson}, \citenamefont
  {Ryan}, \citenamefont {Chow}, \citenamefont {Merkel}, \citenamefont
  {da~Silva}, \citenamefont {Keefe}, \citenamefont {Rothwell}, \citenamefont
  {Ohki}, \citenamefont {Ketchen},\ and\ \citenamefont
  {Steffen}}]{Magesan2012b}%
  \BibitemOpen
  \bibfield  {author} {\bibinfo {author} {\bibfnamefont {E.}~\bibnamefont
  {Magesan}}, \bibinfo {author} {\bibfnamefont {J.~M.}\ \bibnamefont
  {Gambetta}}, \bibinfo {author} {\bibfnamefont {B.~R.}\ \bibnamefont
  {Johnson}}, \bibinfo {author} {\bibfnamefont {C.~A.}\ \bibnamefont {Ryan}},
  \bibinfo {author} {\bibfnamefont {J.~M.}\ \bibnamefont {Chow}}, \bibinfo
  {author} {\bibfnamefont {S.~T.}\ \bibnamefont {Merkel}}, \bibinfo {author}
  {\bibfnamefont {M.~P.}\ \bibnamefont {da~Silva}}, \bibinfo {author}
  {\bibfnamefont {G.~A.}\ \bibnamefont {Keefe}}, \bibinfo {author}
  {\bibfnamefont {M.~B.}\ \bibnamefont {Rothwell}}, \bibinfo {author}
  {\bibfnamefont {T.~A.}\ \bibnamefont {Ohki}}, \bibinfo {author}
  {\bibfnamefont {M.~B.}\ \bibnamefont {Ketchen}}, \ and\ \bibinfo {author}
  {\bibfnamefont {M.}~\bibnamefont {Steffen}},\ }\href {\doibase
  10.1103/PhysRevLett.109.080505} {\bibfield  {journal} {\bibinfo  {journal}
  {Physical Review Letters}\ }\textbf {\bibinfo {volume} {109}},\ \bibinfo
  {pages} {080505} (\bibinfo {year} {2012})},\ \Eprint
  {http://arxiv.org/abs/1203.4550} {arXiv:1203.4550} \BibitemShut {NoStop}%
\bibitem [{\citenamefont {Wallman}\ \emph
  {et~al.}(2015{\natexlab{a}})\citenamefont {Wallman}, \citenamefont {Granade},
  \citenamefont {Harper},\ and\ \citenamefont {Flammia}}]{Wallman2015}%
  \BibitemOpen
  \bibfield  {author} {\bibinfo {author} {\bibfnamefont {J.}~\bibnamefont
  {Wallman}}, \bibinfo {author} {\bibfnamefont {C.}~\bibnamefont {Granade}},
  \bibinfo {author} {\bibfnamefont {R.}~\bibnamefont {Harper}}, \ and\ \bibinfo
  {author} {\bibfnamefont {S.~T.}\ \bibnamefont {Flammia}},\ }\href
  {http://stacks.iop.org/1367-2630/17/i=11/a=113020} {\bibfield  {journal}
  {\bibinfo  {journal} {New Journal of Physics}\ }\textbf {\bibinfo {volume}
  {17}},\ \bibinfo {pages} {113020} (\bibinfo {year}
  {2015}{\natexlab{a}})}\BibitemShut {NoStop}%
\bibitem [{\citenamefont {Wallman}\ \emph
  {et~al.}(2015{\natexlab{b}})\citenamefont {Wallman}, \citenamefont
  {Barnhill},\ and\ \citenamefont {Emerson}}]{Wallman2015b}%
  \BibitemOpen
  \bibfield  {author} {\bibinfo {author} {\bibfnamefont {J.~J.}\ \bibnamefont
  {Wallman}}, \bibinfo {author} {\bibfnamefont {M.}~\bibnamefont {Barnhill}}, \
  and\ \bibinfo {author} {\bibfnamefont {J.}~\bibnamefont {Emerson}},\ }\href
  {\doibase 10.1103/PhysRevLett.115.060501} {\bibfield  {journal} {\bibinfo
  {journal} {Physical Review Letters}\ }\textbf {\bibinfo {volume} {115}},\
  \bibinfo {pages} {060501} (\bibinfo {year} {2015}{\natexlab{b}})}\BibitemShut
  {NoStop}%
\bibitem [{\citenamefont {Wallman}\ and\ \citenamefont
  {Emerson}(2016)}]{Wallman2016}%
  \BibitemOpen
  \bibfield  {author} {\bibinfo {author} {\bibfnamefont {J.~J.}\ \bibnamefont
  {Wallman}}\ and\ \bibinfo {author} {\bibfnamefont {J.}~\bibnamefont
  {Emerson}},\ }\href {\doibase 10.1103/PhysRevA.94.052325} {\bibfield
  {journal} {\bibinfo  {journal} {Physical Review A}\ }\textbf {\bibinfo
  {volume} {94}},\ \bibinfo {pages} {052325} (\bibinfo {year}
  {2016})}\BibitemShut {NoStop}%
\bibitem [{\citenamefont {{Proctor}}\ \emph
  {et~al.}(2017{\natexlab{a}})\citenamefont {{Proctor}}, \citenamefont
  {{Rudinger}}, \citenamefont {{Young}}, \citenamefont {{Sarovar}},\ and\
  \citenamefont {{Blume-Kohout}}}]{Proctor2017}%
  \BibitemOpen
  \bibfield  {author} {\bibinfo {author} {\bibfnamefont {T.}~\bibnamefont
  {{Proctor}}}, \bibinfo {author} {\bibfnamefont {K.}~\bibnamefont
  {{Rudinger}}}, \bibinfo {author} {\bibfnamefont {K.}~\bibnamefont {{Young}}},
  \bibinfo {author} {\bibfnamefont {M.}~\bibnamefont {{Sarovar}}}, \ and\
  \bibinfo {author} {\bibfnamefont {R.}~\bibnamefont {{Blume-Kohout}}},\ }\href
  {\doibase 10.1103/PhysRevLett.119.130502} {\bibfield  {journal} {\bibinfo
  {journal} {\prl}\ }\textbf {\bibinfo {volume} {119}},\ \bibinfo {eid}
  {130502} (\bibinfo {year} {2017}{\natexlab{a}})}\BibitemShut {NoStop}%
\bibitem [{\citenamefont {Wallman}(2018)}]{Wallman2017}%
  \BibitemOpen
  \bibfield  {author} {\bibinfo {author} {\bibfnamefont {J.~J.}\ \bibnamefont
  {Wallman}},\ }\href {\doibase 10.22331/q-2018-01-29-47} {\bibfield  {journal}
  {\bibinfo  {journal} {{Quantum}}\ }\textbf {\bibinfo {volume} {2}},\ \bibinfo
  {pages} {47} (\bibinfo {year} {2018})}\BibitemShut {NoStop}%
\bibitem [{\citenamefont {{Carignan-Dugas}}\ \emph {et~al.}(2018)\citenamefont
  {{Carignan-Dugas}}, \citenamefont {{Boone}}, \citenamefont {{Wallman}},\ and\
  \citenamefont {{Emerson}}}]{Dugas2018}%
  \BibitemOpen
  \bibfield  {author} {\bibinfo {author} {\bibfnamefont {A.}~\bibnamefont
  {{Carignan-Dugas}}}, \bibinfo {author} {\bibfnamefont {K.}~\bibnamefont
  {{Boone}}}, \bibinfo {author} {\bibfnamefont {J.~J.}\ \bibnamefont
  {{Wallman}}}, \ and\ \bibinfo {author} {\bibfnamefont {J.}~\bibnamefont
  {{Emerson}}},\ }\href {\doibase 10.1088/1367-2630/aadcc7} {\bibfield
  {journal} {\bibinfo  {journal} {New Journal of Physics}\ }\textbf {\bibinfo
  {volume} {20}},\ \bibinfo {eid} {092001} (\bibinfo {year}
  {2018})}\BibitemShut {NoStop}%
\bibitem [{\citenamefont {{Harper}}\ \emph {et~al.}(2019)\citenamefont
  {{Harper}}, \citenamefont {{Hincks}}, \citenamefont {{Ferrie}}, \citenamefont
  {{Flammia}},\ and\ \citenamefont {{Wallman}}}]{Harper2019}%
  \BibitemOpen
  \bibfield  {author} {\bibinfo {author} {\bibfnamefont {R.}~\bibnamefont
  {{Harper}}}, \bibinfo {author} {\bibfnamefont {I.}~\bibnamefont {{Hincks}}},
  \bibinfo {author} {\bibfnamefont {C.}~\bibnamefont {{Ferrie}}}, \bibinfo
  {author} {\bibfnamefont {S.~T.}\ \bibnamefont {{Flammia}}}, \ and\ \bibinfo
  {author} {\bibfnamefont {J.~J.}\ \bibnamefont {{Wallman}}},\ }\href@noop {}
  {\bibfield  {journal} {\bibinfo  {journal} {arXiv e-prints}\ ,\ \bibinfo
  {eid} {arXiv:1901.00535}} (\bibinfo {year} {2019})},\ \Eprint
  {http://arxiv.org/abs/1901.00535} {arXiv:1901.00535 [quant-ph]} \BibitemShut
  {NoStop}%
\bibitem [{\citenamefont {{Dirkse}}\ \emph {et~al.}(2019)\citenamefont
  {{Dirkse}}, \citenamefont {{Helsen}},\ and\ \citenamefont
  {{Wehner}}}]{Dirkse2019}%
  \BibitemOpen
  \bibfield  {author} {\bibinfo {author} {\bibfnamefont {B.}~\bibnamefont
  {{Dirkse}}}, \bibinfo {author} {\bibfnamefont {J.}~\bibnamefont {{Helsen}}},
  \ and\ \bibinfo {author} {\bibfnamefont {S.}~\bibnamefont {{Wehner}}},\
  }\href {\doibase 10.1103/PhysRevA.99.012315} {\bibfield  {journal} {\bibinfo
  {journal} {Physical Review A}\ }\textbf {\bibinfo {volume} {99}},\ \bibinfo
  {eid} {012315} (\bibinfo {year} {2019})},\ \Eprint
  {http://arxiv.org/abs/1808.00850} {arXiv:1808.00850 [quant-ph]} \BibitemShut
  {NoStop}%
\bibitem [{\citenamefont {{Merkel}}\ \emph {et~al.}(2018)\citenamefont
  {{Merkel}}, \citenamefont {{Pritchett}},\ and\ \citenamefont
  {{Fong}}}]{Merkel2018}%
  \BibitemOpen
  \bibfield  {author} {\bibinfo {author} {\bibfnamefont {S.~T.}\ \bibnamefont
  {{Merkel}}}, \bibinfo {author} {\bibfnamefont {E.~J.}\ \bibnamefont
  {{Pritchett}}}, \ and\ \bibinfo {author} {\bibfnamefont {B.~H.}\ \bibnamefont
  {{Fong}}},\ }\href@noop {} {\bibfield  {journal} {\bibinfo  {journal} {arXiv
  e-prints}\ ,\ \bibinfo {eid} {arXiv:1804.05951}} (\bibinfo {year} {2018})},\
  \Eprint {http://arxiv.org/abs/1804.05951} {arXiv:1804.05951 [quant-ph]}
  \BibitemShut {NoStop}%
\bibitem [{\citenamefont {{Proctor}}\ \emph
  {et~al.}(2017{\natexlab{b}})\citenamefont {{Proctor}}, \citenamefont
  {{Rudinger}}, \citenamefont {{Young}}, \citenamefont {{Sarovar}},\ and\
  \citenamefont {{Blume-Kohout}}}]{Proctor2017v1}%
  \BibitemOpen
  \bibfield  {author} {\bibinfo {author} {\bibfnamefont {T.}~\bibnamefont
  {{Proctor}}}, \bibinfo {author} {\bibfnamefont {K.}~\bibnamefont
  {{Rudinger}}}, \bibinfo {author} {\bibfnamefont {K.}~\bibnamefont {{Young}}},
  \bibinfo {author} {\bibfnamefont {M.}~\bibnamefont {{Sarovar}}}, \ and\
  \bibinfo {author} {\bibfnamefont {R.}~\bibnamefont {{Blume-Kohout}}},\
  }\href@noop {} {\bibfield  {journal} {\bibinfo  {journal} {arXiv e-prints}\
  ,\ \bibinfo {eid} {arXiv:1702.01853v1}} (\bibinfo {year}
  {2017}{\natexlab{b}})},\ \Eprint {http://arxiv.org/abs/1702.01853v1}
  {arXiv:1702.01853v1 [quant-ph]} \BibitemShut {NoStop}%
\bibitem [{\citenamefont {{Qi}}\ and\ \citenamefont {{Khoon
  Ng}}(2018)}]{Qi2018}%
  \BibitemOpen
  \bibfield  {author} {\bibinfo {author} {\bibfnamefont {J.}~\bibnamefont
  {{Qi}}}\ and\ \bibinfo {author} {\bibfnamefont {H.}~\bibnamefont {{Khoon
  Ng}}},\ }\href@noop {} {\bibfield  {journal} {\bibinfo  {journal} {arXiv
  e-prints}\ ,\ \bibinfo {eid} {arXiv:1805.10622}} (\bibinfo {year} {2018})},\
  \Eprint {http://arxiv.org/abs/1805.10622} {arXiv:1805.10622 [quant-ph]}
  \BibitemShut {NoStop}%
\bibitem [{\citenamefont {{Veldhorst}}\ \emph {et~al.}(2014)\citenamefont
  {{Veldhorst}}, \citenamefont {{Hwang}}, \citenamefont {{Yang}}, \citenamefont
  {{Leenstra}}, \citenamefont {{de Ronde}}, \citenamefont {{Dehollain}},
  \citenamefont {{Muhonen}}, \citenamefont {{Hudson}}, \citenamefont {{Itoh}},
  \citenamefont {{Morello}},\ and\ \citenamefont {{Dzurak}}}]{Veldhorst2014}%
  \BibitemOpen
  \bibfield  {author} {\bibinfo {author} {\bibfnamefont {M.}~\bibnamefont
  {{Veldhorst}}}, \bibinfo {author} {\bibfnamefont {J.~C.~C.}\ \bibnamefont
  {{Hwang}}}, \bibinfo {author} {\bibfnamefont {C.~H.}\ \bibnamefont {{Yang}}},
  \bibinfo {author} {\bibfnamefont {A.~W.}\ \bibnamefont {{Leenstra}}},
  \bibinfo {author} {\bibfnamefont {B.}~\bibnamefont {{de Ronde}}}, \bibinfo
  {author} {\bibfnamefont {J.~P.}\ \bibnamefont {{Dehollain}}}, \bibinfo
  {author} {\bibfnamefont {J.~T.}\ \bibnamefont {{Muhonen}}}, \bibinfo {author}
  {\bibfnamefont {F.~E.}\ \bibnamefont {{Hudson}}}, \bibinfo {author}
  {\bibfnamefont {K.~M.}\ \bibnamefont {{Itoh}}}, \bibinfo {author}
  {\bibfnamefont {A.}~\bibnamefont {{Morello}}}, \ and\ \bibinfo {author}
  {\bibfnamefont {A.~S.}\ \bibnamefont {{Dzurak}}},\ }\href {\doibase
  10.1038/nnano.2014.216} {\bibfield  {journal} {\bibinfo  {journal} {Nature
  Nanotechnology}\ }\textbf {\bibinfo {volume} {9}},\ \bibinfo {pages} {981}
  (\bibinfo {year} {2014})},\ \Eprint {http://arxiv.org/abs/1407.1950}
  {arXiv:1407.1950 [cond-mat.mes-hall]} \BibitemShut {NoStop}%
\bibitem [{\citenamefont {{Veldhorst}}\ \emph {et~al.}(2015)\citenamefont
  {{Veldhorst}}, \citenamefont {{Ruskov}}, \citenamefont {{Yang}},
  \citenamefont {{Hwang}}, \citenamefont {{Hudson}}, \citenamefont
  {{Flatt{\'e}}}, \citenamefont {{Tahan}}, \citenamefont {{Itoh}},
  \citenamefont {{Morello}},\ and\ \citenamefont {{Dzurak}}}]{Veldhorst2015}%
  \BibitemOpen
  \bibfield  {author} {\bibinfo {author} {\bibfnamefont {M.}~\bibnamefont
  {{Veldhorst}}}, \bibinfo {author} {\bibfnamefont {R.}~\bibnamefont
  {{Ruskov}}}, \bibinfo {author} {\bibfnamefont {C.~H.}\ \bibnamefont
  {{Yang}}}, \bibinfo {author} {\bibfnamefont {J.~C.~C.}\ \bibnamefont
  {{Hwang}}}, \bibinfo {author} {\bibfnamefont {F.~E.}\ \bibnamefont
  {{Hudson}}}, \bibinfo {author} {\bibfnamefont {M.~E.}\ \bibnamefont
  {{Flatt{\'e}}}}, \bibinfo {author} {\bibfnamefont {C.}~\bibnamefont
  {{Tahan}}}, \bibinfo {author} {\bibfnamefont {K.~M.}\ \bibnamefont {{Itoh}}},
  \bibinfo {author} {\bibfnamefont {A.}~\bibnamefont {{Morello}}}, \ and\
  \bibinfo {author} {\bibfnamefont {A.~S.}\ \bibnamefont {{Dzurak}}},\ }\href
  {\doibase 10.1103/PhysRevB.92.201401} {\bibfield  {journal} {\bibinfo
  {journal} {Physical Review B}\ }\textbf {\bibinfo {volume} {92}},\ \bibinfo
  {eid} {201401} (\bibinfo {year} {2015})},\ \Eprint
  {http://arxiv.org/abs/1505.01213} {arXiv:1505.01213 [cond-mat.mes-hall]}
  \BibitemShut {NoStop}%
\bibitem [{\citenamefont {{Barends}}\ \emph {et~al.}(2015)\citenamefont
  {{Barends}}, \citenamefont {{Lamata}}, \citenamefont {{Kelly}}, \citenamefont
  {{Garc{\'\i}a-{\'A}lvarez}}, \citenamefont {{Fowler}}, \citenamefont
  {{Megrant}}, \citenamefont {{Jeffrey}}, \citenamefont {{White}},
  \citenamefont {{Sank}}, \citenamefont {{Mutus}}, \citenamefont {{Campbell}},
  \citenamefont {{Chen}}, \citenamefont {{Chen}}, \citenamefont {{Chiaro}},
  \citenamefont {{Dunsworth}}, \citenamefont {{Hoi}}, \citenamefont {{Neill}},
  \citenamefont {{O'Malley}}, \citenamefont {{Quintana}}, \citenamefont
  {{Roushan}}, \citenamefont {{Vainsencher}}, \citenamefont {{Wenner}},
  \citenamefont {{Solano}},\ and\ \citenamefont {{Martinis}}}]{Barends2015}%
  \BibitemOpen
  \bibfield  {author} {\bibinfo {author} {\bibfnamefont {R.}~\bibnamefont
  {{Barends}}}, \bibinfo {author} {\bibfnamefont {L.}~\bibnamefont {{Lamata}}},
  \bibinfo {author} {\bibfnamefont {J.}~\bibnamefont {{Kelly}}}, \bibinfo
  {author} {\bibfnamefont {L.}~\bibnamefont {{Garc{\'\i}a-{\'A}lvarez}}},
  \bibinfo {author} {\bibfnamefont {A.~G.}\ \bibnamefont {{Fowler}}}, \bibinfo
  {author} {\bibfnamefont {A.}~\bibnamefont {{Megrant}}}, \bibinfo {author}
  {\bibfnamefont {E.}~\bibnamefont {{Jeffrey}}}, \bibinfo {author}
  {\bibfnamefont {T.~C.}\ \bibnamefont {{White}}}, \bibinfo {author}
  {\bibfnamefont {D.}~\bibnamefont {{Sank}}}, \bibinfo {author} {\bibfnamefont
  {J.~Y.}\ \bibnamefont {{Mutus}}}, \bibinfo {author} {\bibfnamefont
  {B.}~\bibnamefont {{Campbell}}}, \bibinfo {author} {\bibfnamefont
  {Y.}~\bibnamefont {{Chen}}}, \bibinfo {author} {\bibfnamefont
  {Z.}~\bibnamefont {{Chen}}}, \bibinfo {author} {\bibfnamefont
  {B.}~\bibnamefont {{Chiaro}}}, \bibinfo {author} {\bibfnamefont
  {A.}~\bibnamefont {{Dunsworth}}}, \bibinfo {author} {\bibfnamefont {I.~C.}\
  \bibnamefont {{Hoi}}}, \bibinfo {author} {\bibfnamefont {C.}~\bibnamefont
  {{Neill}}}, \bibinfo {author} {\bibfnamefont {P.~J.~J.}\ \bibnamefont
  {{O'Malley}}}, \bibinfo {author} {\bibfnamefont {C.}~\bibnamefont
  {{Quintana}}}, \bibinfo {author} {\bibfnamefont {P.}~\bibnamefont
  {{Roushan}}}, \bibinfo {author} {\bibfnamefont {A.}~\bibnamefont
  {{Vainsencher}}}, \bibinfo {author} {\bibfnamefont {J.}~\bibnamefont
  {{Wenner}}}, \bibinfo {author} {\bibfnamefont {E.}~\bibnamefont {{Solano}}},
  \ and\ \bibinfo {author} {\bibfnamefont {J.~M.}\ \bibnamefont {{Martinis}}},\
  }\href {\doibase 10.1038/ncomms8654} {\bibfield  {journal} {\bibinfo
  {journal} {Nature Communications}\ }\textbf {\bibinfo {volume} {6}},\
  \bibinfo {eid} {7654} (\bibinfo {year} {2015})},\ \Eprint
  {http://arxiv.org/abs/1501.07703} {arXiv:1501.07703 [quant-ph]} \BibitemShut
  {NoStop}%
\bibitem [{\citenamefont {{Takeda}}\ \emph {et~al.}(2016)\citenamefont
  {{Takeda}}, \citenamefont {{Kamioka}}, \citenamefont {{Otsuka}},
  \citenamefont {{Yoneda}}, \citenamefont {{Nakajima}}, \citenamefont
  {{Delbecq}}, \citenamefont {{Amaha}}, \citenamefont {{Allison}},
  \citenamefont {{Kodera}}, \citenamefont {{Oda}},\ and\ \citenamefont
  {{Tarucha}}}]{Takeda2016}%
  \BibitemOpen
  \bibfield  {author} {\bibinfo {author} {\bibfnamefont {K.}~\bibnamefont
  {{Takeda}}}, \bibinfo {author} {\bibfnamefont {J.}~\bibnamefont {{Kamioka}}},
  \bibinfo {author} {\bibfnamefont {T.}~\bibnamefont {{Otsuka}}}, \bibinfo
  {author} {\bibfnamefont {J.}~\bibnamefont {{Yoneda}}}, \bibinfo {author}
  {\bibfnamefont {T.}~\bibnamefont {{Nakajima}}}, \bibinfo {author}
  {\bibfnamefont {M.~R.}\ \bibnamefont {{Delbecq}}}, \bibinfo {author}
  {\bibfnamefont {S.}~\bibnamefont {{Amaha}}}, \bibinfo {author} {\bibfnamefont
  {G.}~\bibnamefont {{Allison}}}, \bibinfo {author} {\bibfnamefont
  {T.}~\bibnamefont {{Kodera}}}, \bibinfo {author} {\bibfnamefont
  {S.}~\bibnamefont {{Oda}}}, \ and\ \bibinfo {author} {\bibfnamefont
  {S.}~\bibnamefont {{Tarucha}}},\ }\href {\doibase 10.1126/sciadv.1600694}
  {\bibfield  {journal} {\bibinfo  {journal} {Science Advances}\ }\textbf
  {\bibinfo {volume} {2}},\ \bibinfo {pages} {e1600694} (\bibinfo {year}
  {2016})},\ \Eprint {http://arxiv.org/abs/1602.07833} {arXiv:1602.07833
  [cond-mat.mes-hall]} \BibitemShut {NoStop}%
\bibitem [{\citenamefont {{McKay}}\ \emph
  {et~al.}(2017{\natexlab{b}})\citenamefont {{McKay}}, \citenamefont {{Wood}},
  \citenamefont {{Sheldon}}, \citenamefont {{Chow}},\ and\ \citenamefont
  {{Gambetta}}}]{McKay2017b}%
  \BibitemOpen
  \bibfield  {author} {\bibinfo {author} {\bibfnamefont {D.~C.}\ \bibnamefont
  {{McKay}}}, \bibinfo {author} {\bibfnamefont {C.~J.}\ \bibnamefont {{Wood}}},
  \bibinfo {author} {\bibfnamefont {S.}~\bibnamefont {{Sheldon}}}, \bibinfo
  {author} {\bibfnamefont {J.~M.}\ \bibnamefont {{Chow}}}, \ and\ \bibinfo
  {author} {\bibfnamefont {J.~M.}\ \bibnamefont {{Gambetta}}},\ }\href
  {\doibase 10.1103/PhysRevA.96.022330} {\bibfield  {journal} {\bibinfo
  {journal} {Physical Review A}\ }\textbf {\bibinfo {volume} {96}},\ \bibinfo
  {eid} {022330} (\bibinfo {year} {2017}{\natexlab{b}})},\ \Eprint
  {http://arxiv.org/abs/1612.00858} {arXiv:1612.00858 [quant-ph]} \BibitemShut
  {NoStop}%
\bibitem [{\citenamefont {{Nichol}}\ \emph {et~al.}(2017)\citenamefont
  {{Nichol}}, \citenamefont {{Orona}}, \citenamefont {{Harvey}}, \citenamefont
  {{Fallahi}}, \citenamefont {{Gardner}}, \citenamefont {{Manfra}},\ and\
  \citenamefont {{Yacoby}}}]{Nichol2017}%
  \BibitemOpen
  \bibfield  {author} {\bibinfo {author} {\bibfnamefont {J.~M.}\ \bibnamefont
  {{Nichol}}}, \bibinfo {author} {\bibfnamefont {L.~A.}\ \bibnamefont
  {{Orona}}}, \bibinfo {author} {\bibfnamefont {S.~P.}\ \bibnamefont
  {{Harvey}}}, \bibinfo {author} {\bibfnamefont {S.}~\bibnamefont {{Fallahi}}},
  \bibinfo {author} {\bibfnamefont {G.~C.}\ \bibnamefont {{Gardner}}}, \bibinfo
  {author} {\bibfnamefont {M.~J.}\ \bibnamefont {{Manfra}}}, \ and\ \bibinfo
  {author} {\bibfnamefont {A.}~\bibnamefont {{Yacoby}}},\ }\href {\doibase
  10.1038/s41534-016-0003-1} {\bibfield  {journal} {\bibinfo  {journal} {npj
  Quantum Information}\ }\textbf {\bibinfo {volume} {3}},\ \bibinfo {eid} {3}
  (\bibinfo {year} {2017})},\ \Eprint {http://arxiv.org/abs/1608.04258}
  {arXiv:1608.04258 [cond-mat.mes-hall]} \BibitemShut {NoStop}%
\bibitem [{\citenamefont {{Chan}}\ \emph {et~al.}(2018)\citenamefont {{Chan}},
  \citenamefont {{Huang}}, \citenamefont {{Yang}}, \citenamefont {{Hwang}},
  \citenamefont {{Hensen}}, \citenamefont {{Tanttu}}, \citenamefont {{Hudson}},
  \citenamefont {{Itoh}}, \citenamefont {{Laucht}}, \citenamefont {{Morello}},\
  and\ \citenamefont {{Dzurak}}}]{Chan2018}%
  \BibitemOpen
  \bibfield  {author} {\bibinfo {author} {\bibfnamefont {K.~W.}\ \bibnamefont
  {{Chan}}}, \bibinfo {author} {\bibfnamefont {W.}~\bibnamefont {{Huang}}},
  \bibinfo {author} {\bibfnamefont {C.~H.}\ \bibnamefont {{Yang}}}, \bibinfo
  {author} {\bibfnamefont {J.~C.~C.}\ \bibnamefont {{Hwang}}}, \bibinfo
  {author} {\bibfnamefont {B.}~\bibnamefont {{Hensen}}}, \bibinfo {author}
  {\bibfnamefont {T.}~\bibnamefont {{Tanttu}}}, \bibinfo {author}
  {\bibfnamefont {F.~E.}\ \bibnamefont {{Hudson}}}, \bibinfo {author}
  {\bibfnamefont {K.~M.}\ \bibnamefont {{Itoh}}}, \bibinfo {author}
  {\bibfnamefont {A.}~\bibnamefont {{Laucht}}}, \bibinfo {author}
  {\bibfnamefont {A.}~\bibnamefont {{Morello}}}, \ and\ \bibinfo {author}
  {\bibfnamefont {A.~S.}\ \bibnamefont {{Dzurak}}},\ }\href {\doibase
  10.1103/PhysRevApplied.10.044017} {\bibfield  {journal} {\bibinfo  {journal}
  {Physical Review Applied}\ }\textbf {\bibinfo {volume} {10}},\ \bibinfo {eid}
  {044017} (\bibinfo {year} {2018})},\ \Eprint
  {http://arxiv.org/abs/1803.01609} {arXiv:1803.01609 [quant-ph]} \BibitemShut
  {NoStop}%
\bibitem [{\citenamefont {{Caldwell}}\ \emph {et~al.}(2018)\citenamefont
  {{Caldwell}}, \citenamefont {{Didier}}, \citenamefont {{Ryan}}, \citenamefont
  {{Sete}}, \citenamefont {{Hudson}}, \citenamefont {{Karalekas}},
  \citenamefont {{Manenti}}, \citenamefont {{da Silva}}, \citenamefont
  {{Sinclair}}, \citenamefont {{Acala}}, \citenamefont {{Alidoust}},
  \citenamefont {{Angeles}}, \citenamefont {{Bestwick}}, \citenamefont
  {{Block}}, \citenamefont {{Bloom}}, \citenamefont {{Bradley}}, \citenamefont
  {{Bui}}, \citenamefont {{Capelluto}}, \citenamefont {{Chilcott}},
  \citenamefont {{Cordova}}, \citenamefont {{Crossman}}, \citenamefont
  {{Curtis}}, \citenamefont {{Deshpande}}, \citenamefont {{Bouayadi}},
  \citenamefont {{Girshovich}}, \citenamefont {{Hong}}, \citenamefont
  {{Kuang}}, \citenamefont {{Lenihan}}, \citenamefont {{Manning}},
  \citenamefont {{Marchenkov}}, \citenamefont {{Marshall}}, \citenamefont
  {{Maydra}}, \citenamefont {{Mohan}}, \citenamefont {{O'Brien}}, \citenamefont
  {{Osborn}}, \citenamefont {{Otterbach}}, \citenamefont {{Papageorge}},
  \citenamefont {{Paquette}}, \citenamefont {{Pelstring}}, \citenamefont
  {{Polloreno}}, \citenamefont {{Prawiroatmodjo}}, \citenamefont {{Rawat}},
  \citenamefont {{Reagor}}, \citenamefont {{Renzas}}, \citenamefont {{Rubin}},
  \citenamefont {{Russell}}, \citenamefont {{Rust}}, \citenamefont
  {{Scarabelli}}, \citenamefont {{Scheer}}, \citenamefont {{Selvanayagam}},
  \citenamefont {{Smith}}, \citenamefont {{Staley}}, \citenamefont {{Suska}},
  \citenamefont {{Tezak}}, \citenamefont {{Thompson}}, \citenamefont {{To}},
  \citenamefont {{Vahidpour}}, \citenamefont {{Vodrahalli}}, \citenamefont
  {{Whyland}}, \citenamefont {{Yadav}}, \citenamefont {{Zeng}},\ and\
  \citenamefont {{Rigetti}}}]{Caldwell2018}%
  \BibitemOpen
  \bibfield  {author} {\bibinfo {author} {\bibfnamefont {S.~A.}\ \bibnamefont
  {{Caldwell}}}, \bibinfo {author} {\bibfnamefont {N.}~\bibnamefont
  {{Didier}}}, \bibinfo {author} {\bibfnamefont {C.~A.}\ \bibnamefont
  {{Ryan}}}, \bibinfo {author} {\bibfnamefont {E.~A.}\ \bibnamefont {{Sete}}},
  \bibinfo {author} {\bibfnamefont {A.}~\bibnamefont {{Hudson}}}, \bibinfo
  {author} {\bibfnamefont {P.}~\bibnamefont {{Karalekas}}}, \bibinfo {author}
  {\bibfnamefont {R.}~\bibnamefont {{Manenti}}}, \bibinfo {author}
  {\bibfnamefont {M.~P.}\ \bibnamefont {{da Silva}}}, \bibinfo {author}
  {\bibfnamefont {R.}~\bibnamefont {{Sinclair}}}, \bibinfo {author}
  {\bibfnamefont {E.}~\bibnamefont {{Acala}}}, \bibinfo {author} {\bibfnamefont
  {N.}~\bibnamefont {{Alidoust}}}, \bibinfo {author} {\bibfnamefont
  {J.}~\bibnamefont {{Angeles}}}, \bibinfo {author} {\bibfnamefont
  {A.}~\bibnamefont {{Bestwick}}}, \bibinfo {author} {\bibfnamefont
  {M.}~\bibnamefont {{Block}}}, \bibinfo {author} {\bibfnamefont
  {B.}~\bibnamefont {{Bloom}}}, \bibinfo {author} {\bibfnamefont
  {A.}~\bibnamefont {{Bradley}}}, \bibinfo {author} {\bibfnamefont
  {C.}~\bibnamefont {{Bui}}}, \bibinfo {author} {\bibfnamefont
  {L.}~\bibnamefont {{Capelluto}}}, \bibinfo {author} {\bibfnamefont
  {R.}~\bibnamefont {{Chilcott}}}, \bibinfo {author} {\bibfnamefont
  {J.}~\bibnamefont {{Cordova}}}, \bibinfo {author} {\bibfnamefont
  {G.}~\bibnamefont {{Crossman}}}, \bibinfo {author} {\bibfnamefont
  {M.}~\bibnamefont {{Curtis}}}, \bibinfo {author} {\bibfnamefont
  {S.}~\bibnamefont {{Deshpande}}}, \bibinfo {author} {\bibfnamefont {T.~E.}\
  \bibnamefont {{Bouayadi}}}, \bibinfo {author} {\bibfnamefont
  {D.}~\bibnamefont {{Girshovich}}}, \bibinfo {author} {\bibfnamefont
  {S.}~\bibnamefont {{Hong}}}, \bibinfo {author} {\bibfnamefont
  {K.}~\bibnamefont {{Kuang}}}, \bibinfo {author} {\bibfnamefont
  {M.}~\bibnamefont {{Lenihan}}}, \bibinfo {author} {\bibfnamefont
  {T.}~\bibnamefont {{Manning}}}, \bibinfo {author} {\bibfnamefont
  {A.}~\bibnamefont {{Marchenkov}}}, \bibinfo {author} {\bibfnamefont
  {J.}~\bibnamefont {{Marshall}}}, \bibinfo {author} {\bibfnamefont
  {R.}~\bibnamefont {{Maydra}}}, \bibinfo {author} {\bibfnamefont
  {Y.}~\bibnamefont {{Mohan}}}, \bibinfo {author} {\bibfnamefont
  {W.}~\bibnamefont {{O'Brien}}}, \bibinfo {author} {\bibfnamefont
  {C.}~\bibnamefont {{Osborn}}}, \bibinfo {author} {\bibfnamefont
  {J.}~\bibnamefont {{Otterbach}}}, \bibinfo {author} {\bibfnamefont
  {A.}~\bibnamefont {{Papageorge}}}, \bibinfo {author} {\bibfnamefont {J.~P.}\
  \bibnamefont {{Paquette}}}, \bibinfo {author} {\bibfnamefont
  {M.}~\bibnamefont {{Pelstring}}}, \bibinfo {author} {\bibfnamefont
  {A.}~\bibnamefont {{Polloreno}}}, \bibinfo {author} {\bibfnamefont
  {G.}~\bibnamefont {{Prawiroatmodjo}}}, \bibinfo {author} {\bibfnamefont
  {V.}~\bibnamefont {{Rawat}}}, \bibinfo {author} {\bibfnamefont
  {M.}~\bibnamefont {{Reagor}}}, \bibinfo {author} {\bibfnamefont
  {R.}~\bibnamefont {{Renzas}}}, \bibinfo {author} {\bibfnamefont
  {N.}~\bibnamefont {{Rubin}}}, \bibinfo {author} {\bibfnamefont
  {D.}~\bibnamefont {{Russell}}}, \bibinfo {author} {\bibfnamefont
  {M.}~\bibnamefont {{Rust}}}, \bibinfo {author} {\bibfnamefont
  {D.}~\bibnamefont {{Scarabelli}}}, \bibinfo {author} {\bibfnamefont
  {M.}~\bibnamefont {{Scheer}}}, \bibinfo {author} {\bibfnamefont
  {M.}~\bibnamefont {{Selvanayagam}}}, \bibinfo {author} {\bibfnamefont
  {R.}~\bibnamefont {{Smith}}}, \bibinfo {author} {\bibfnamefont
  {A.}~\bibnamefont {{Staley}}}, \bibinfo {author} {\bibfnamefont
  {M.}~\bibnamefont {{Suska}}}, \bibinfo {author} {\bibfnamefont
  {N.}~\bibnamefont {{Tezak}}}, \bibinfo {author} {\bibfnamefont {D.~C.}\
  \bibnamefont {{Thompson}}}, \bibinfo {author} {\bibfnamefont {T.~W.}\
  \bibnamefont {{To}}}, \bibinfo {author} {\bibfnamefont {M.}~\bibnamefont
  {{Vahidpour}}}, \bibinfo {author} {\bibfnamefont {N.}~\bibnamefont
  {{Vodrahalli}}}, \bibinfo {author} {\bibfnamefont {T.}~\bibnamefont
  {{Whyland}}}, \bibinfo {author} {\bibfnamefont {K.}~\bibnamefont {{Yadav}}},
  \bibinfo {author} {\bibfnamefont {W.}~\bibnamefont {{Zeng}}}, \ and\ \bibinfo
  {author} {\bibfnamefont {C.}~\bibnamefont {{Rigetti}}},\ }\href {\doibase
  10.1103/PhysRevApplied.10.034050} {\bibfield  {journal} {\bibinfo  {journal}
  {Physical Review Applied}\ }\textbf {\bibinfo {volume} {10}},\ \bibinfo {eid}
  {034050} (\bibinfo {year} {2018})},\ \Eprint
  {http://arxiv.org/abs/1706.06562} {arXiv:1706.06562 [quant-ph]} \BibitemShut
  {NoStop}%
\bibitem [{\citenamefont {{Wang}}\ \emph
  {et~al.}(2018{\natexlab{a}})\citenamefont {{Wang}}, \citenamefont {{Zhang}},
  \citenamefont {{Xiang}}, \citenamefont {{Jia}}, \citenamefont {{Duan}},
  \citenamefont {{Cai}}, \citenamefont {{Gong}}, \citenamefont {{Zong}},
  \citenamefont {{Wu}}, \citenamefont {{Wu}}, \citenamefont {{Sun}},
  \citenamefont {{Yin}},\ and\ \citenamefont {{Guo}}}]{Wang2018}%
  \BibitemOpen
  \bibfield  {author} {\bibinfo {author} {\bibfnamefont {T.}~\bibnamefont
  {{Wang}}}, \bibinfo {author} {\bibfnamefont {Z.}~\bibnamefont {{Zhang}}},
  \bibinfo {author} {\bibfnamefont {L.}~\bibnamefont {{Xiang}}}, \bibinfo
  {author} {\bibfnamefont {Z.}~\bibnamefont {{Jia}}}, \bibinfo {author}
  {\bibfnamefont {P.}~\bibnamefont {{Duan}}}, \bibinfo {author} {\bibfnamefont
  {W.}~\bibnamefont {{Cai}}}, \bibinfo {author} {\bibfnamefont
  {Z.}~\bibnamefont {{Gong}}}, \bibinfo {author} {\bibfnamefont
  {Z.}~\bibnamefont {{Zong}}}, \bibinfo {author} {\bibfnamefont
  {M.}~\bibnamefont {{Wu}}}, \bibinfo {author} {\bibfnamefont {J.}~\bibnamefont
  {{Wu}}}, \bibinfo {author} {\bibfnamefont {L.}~\bibnamefont {{Sun}}},
  \bibinfo {author} {\bibfnamefont {Y.}~\bibnamefont {{Yin}}}, \ and\ \bibinfo
  {author} {\bibfnamefont {G.}~\bibnamefont {{Guo}}},\ }\href {\doibase
  10.1088/1367-2630/aac9e7} {\bibfield  {journal} {\bibinfo  {journal} {New
  Journal of Physics}\ }\textbf {\bibinfo {volume} {20}},\ \bibinfo {eid}
  {065003} (\bibinfo {year} {2018}{\natexlab{a}})},\ \Eprint
  {http://arxiv.org/abs/1804.08247} {arXiv:1804.08247 [quant-ph]} \BibitemShut
  {NoStop}%
\bibitem [{\citenamefont {{Wang}}\ \emph
  {et~al.}(2018{\natexlab{b}})\citenamefont {{Wang}}, \citenamefont {{Zhang}},
  \citenamefont {{Xiang}}, \citenamefont {{Jia}}, \citenamefont {{Duan}},
  \citenamefont {{Zong}}, \citenamefont {{Sun}}, \citenamefont {{Dong}},
  \citenamefont {{Wu}}, \citenamefont {{Yin}},\ and\ \citenamefont
  {{Guo}}}]{Wang2018b}%
  \BibitemOpen
  \bibfield  {author} {\bibinfo {author} {\bibfnamefont {T.}~\bibnamefont
  {{Wang}}}, \bibinfo {author} {\bibfnamefont {Z.}~\bibnamefont {{Zhang}}},
  \bibinfo {author} {\bibfnamefont {L.}~\bibnamefont {{Xiang}}}, \bibinfo
  {author} {\bibfnamefont {Z.}~\bibnamefont {{Jia}}}, \bibinfo {author}
  {\bibfnamefont {P.}~\bibnamefont {{Duan}}}, \bibinfo {author} {\bibfnamefont
  {Z.}~\bibnamefont {{Zong}}}, \bibinfo {author} {\bibfnamefont
  {Z.}~\bibnamefont {{Sun}}}, \bibinfo {author} {\bibfnamefont
  {Z.}~\bibnamefont {{Dong}}}, \bibinfo {author} {\bibfnamefont
  {J.}~\bibnamefont {{Wu}}}, \bibinfo {author} {\bibfnamefont {Y.}~\bibnamefont
  {{Yin}}}, \ and\ \bibinfo {author} {\bibfnamefont {G.}~\bibnamefont
  {{Guo}}},\ }\href@noop {} {\bibfield  {journal} {\bibinfo  {journal} {arXiv
  e-prints}\ ,\ \bibinfo {eid} {arXiv:1811.08096}} (\bibinfo {year}
  {2018}{\natexlab{b}})},\ \Eprint {http://arxiv.org/abs/1811.08096}
  {arXiv:1811.08096 [quant-ph]} \BibitemShut {NoStop}%
\bibitem [{\citenamefont {{Yoneda}}\ \emph {et~al.}(2018)\citenamefont
  {{Yoneda}}, \citenamefont {{Takeda}}, \citenamefont {{Otsuka}}, \citenamefont
  {{Nakajima}}, \citenamefont {{Delbecq}}, \citenamefont {{Allison}},
  \citenamefont {{Honda}}, \citenamefont {{Kodera}}, \citenamefont {{Oda}},
  \citenamefont {{Hoshi}}, \citenamefont {{Usami}}, \citenamefont {{Itoh}},\
  and\ \citenamefont {{Tarucha}}}]{Yoneda2018}%
  \BibitemOpen
  \bibfield  {author} {\bibinfo {author} {\bibfnamefont {J.}~\bibnamefont
  {{Yoneda}}}, \bibinfo {author} {\bibfnamefont {K.}~\bibnamefont {{Takeda}}},
  \bibinfo {author} {\bibfnamefont {T.}~\bibnamefont {{Otsuka}}}, \bibinfo
  {author} {\bibfnamefont {T.}~\bibnamefont {{Nakajima}}}, \bibinfo {author}
  {\bibfnamefont {M.~R.}\ \bibnamefont {{Delbecq}}}, \bibinfo {author}
  {\bibfnamefont {G.}~\bibnamefont {{Allison}}}, \bibinfo {author}
  {\bibfnamefont {T.}~\bibnamefont {{Honda}}}, \bibinfo {author} {\bibfnamefont
  {T.}~\bibnamefont {{Kodera}}}, \bibinfo {author} {\bibfnamefont
  {S.}~\bibnamefont {{Oda}}}, \bibinfo {author} {\bibfnamefont
  {Y.}~\bibnamefont {{Hoshi}}}, \bibinfo {author} {\bibfnamefont
  {N.}~\bibnamefont {{Usami}}}, \bibinfo {author} {\bibfnamefont {K.~M.}\
  \bibnamefont {{Itoh}}}, \ and\ \bibinfo {author} {\bibfnamefont
  {S.}~\bibnamefont {{Tarucha}}},\ }\href {\doibase 10.1038/s41565-017-0014-x}
  {\bibfield  {journal} {\bibinfo  {journal} {Nature Nanotechnology}\ }\textbf
  {\bibinfo {volume} {13}},\ \bibinfo {pages} {102} (\bibinfo {year} {2018})},\
  \Eprint {http://arxiv.org/abs/1708.01454} {arXiv:1708.01454
  [cond-mat.mes-hall]} \BibitemShut {NoStop}%
\bibitem [{\citenamefont {{Zhang}}\ \emph {et~al.}(2018)\citenamefont
  {{Zhang}}, \citenamefont {{Zhao}}, \citenamefont {{Wang}}, \citenamefont
  {{Xiang}}, \citenamefont {{Jia}}, \citenamefont {{Duan}}, \citenamefont
  {{Tong}}, \citenamefont {{Yin}},\ and\ \citenamefont {{Guo}}}]{Zhang2018}%
  \BibitemOpen
  \bibfield  {author} {\bibinfo {author} {\bibfnamefont {Z.}~\bibnamefont
  {{Zhang}}}, \bibinfo {author} {\bibfnamefont {P.~Z.}\ \bibnamefont {{Zhao}}},
  \bibinfo {author} {\bibfnamefont {T.}~\bibnamefont {{Wang}}}, \bibinfo
  {author} {\bibfnamefont {L.}~\bibnamefont {{Xiang}}}, \bibinfo {author}
  {\bibfnamefont {Z.}~\bibnamefont {{Jia}}}, \bibinfo {author} {\bibfnamefont
  {P.}~\bibnamefont {{Duan}}}, \bibinfo {author} {\bibfnamefont {D.~M.}\
  \bibnamefont {{Tong}}}, \bibinfo {author} {\bibfnamefont {Y.}~\bibnamefont
  {{Yin}}}, \ and\ \bibinfo {author} {\bibfnamefont {G.}~\bibnamefont
  {{Guo}}},\ }\href@noop {} {\bibfield  {journal} {\bibinfo  {journal} {arXiv
  e-prints}\ ,\ \bibinfo {eid} {arXiv:1811.06252}} (\bibinfo {year} {2018})},\
  \Eprint {http://arxiv.org/abs/1811.06252} {arXiv:1811.06252 [quant-ph]}
  \BibitemShut {NoStop}%
\bibitem [{\citenamefont {{Carignan-Dugas}}\ \emph {et~al.}(2015)\citenamefont
  {{Carignan-Dugas}}, \citenamefont {{Wallman}},\ and\ \citenamefont
  {{Emerson}}}]{Dugas2015}%
  \BibitemOpen
  \bibfield  {author} {\bibinfo {author} {\bibfnamefont {A.}~\bibnamefont
  {{Carignan-Dugas}}}, \bibinfo {author} {\bibfnamefont {J.~J.}\ \bibnamefont
  {{Wallman}}}, \ and\ \bibinfo {author} {\bibfnamefont {J.}~\bibnamefont
  {{Emerson}}},\ }\href {\doibase 10.1103/PhysRevA.92.060302} {\bibfield
  {journal} {\bibinfo  {journal} {Physical Review A}\ }\textbf {\bibinfo
  {volume} {92}},\ \bibinfo {eid} {060302} (\bibinfo {year} {2015})},\ \Eprint
  {http://arxiv.org/abs/1508.06312} {arXiv:1508.06312 [quant-ph]} \BibitemShut
  {NoStop}%
\bibitem [{\citenamefont {{Cross}}\ \emph {et~al.}(2016)\citenamefont
  {{Cross}}, \citenamefont {{Magesan}}, \citenamefont {{Bishop}}, \citenamefont
  {{Smolin}},\ and\ \citenamefont {{Gambetta}}}]{Cross2016}%
  \BibitemOpen
  \bibfield  {author} {\bibinfo {author} {\bibfnamefont {A.~W.}\ \bibnamefont
  {{Cross}}}, \bibinfo {author} {\bibfnamefont {E.}~\bibnamefont {{Magesan}}},
  \bibinfo {author} {\bibfnamefont {L.~S.}\ \bibnamefont {{Bishop}}}, \bibinfo
  {author} {\bibfnamefont {J.~A.}\ \bibnamefont {{Smolin}}}, \ and\ \bibinfo
  {author} {\bibfnamefont {J.~M.}\ \bibnamefont {{Gambetta}}},\ }\href
  {\doibase 10.1038/npjqi.2016.12} {\bibfield  {journal} {\bibinfo  {journal}
  {npj Quantum Information}\ }\textbf {\bibinfo {volume} {2}},\ \bibinfo {eid}
  {16012} (\bibinfo {year} {2016})},\ \Eprint {http://arxiv.org/abs/1510.02720}
  {arXiv:1510.02720 [quant-ph]} \BibitemShut {NoStop}%
\bibitem [{\citenamefont {{Harper}}\ and\ \citenamefont
  {{Flammia}}(2017)}]{Harper2017}%
  \BibitemOpen
  \bibfield  {author} {\bibinfo {author} {\bibfnamefont {R.}~\bibnamefont
  {{Harper}}}\ and\ \bibinfo {author} {\bibfnamefont {S.~T.}\ \bibnamefont
  {{Flammia}}},\ }\href {\doibase 10.1088/2058-9565/aa5f8d} {\bibfield
  {journal} {\bibinfo  {journal} {Quantum Science and Technology}\ }\textbf
  {\bibinfo {volume} {2}},\ \bibinfo {pages} {015008} (\bibinfo {year}
  {2017})},\ \Eprint {http://arxiv.org/abs/1608.02943} {arXiv:1608.02943
  [quant-ph]} \BibitemShut {NoStop}%
\bibitem [{\citenamefont {{Proctor}}\ \emph {et~al.}(2018)\citenamefont
  {{Proctor}}, \citenamefont {{Carignan-Dugas}}, \citenamefont {{Rudinger}},
  \citenamefont {{Nielsen}}, \citenamefont {{Blume-Kohout}},\ and\
  \citenamefont {{Young}}}]{Proctor2018}%
  \BibitemOpen
  \bibfield  {author} {\bibinfo {author} {\bibfnamefont {T.~J.}\ \bibnamefont
  {{Proctor}}}, \bibinfo {author} {\bibfnamefont {A.}~\bibnamefont
  {{Carignan-Dugas}}}, \bibinfo {author} {\bibfnamefont {K.}~\bibnamefont
  {{Rudinger}}}, \bibinfo {author} {\bibfnamefont {E.}~\bibnamefont
  {{Nielsen}}}, \bibinfo {author} {\bibfnamefont {R.}~\bibnamefont
  {{Blume-Kohout}}}, \ and\ \bibinfo {author} {\bibfnamefont {K.}~\bibnamefont
  {{Young}}},\ }\href@noop {} {\bibfield  {journal} {\bibinfo  {journal} {arXiv
  e-prints}\ ,\ \bibinfo {eid} {arXiv:1807.07975}} (\bibinfo {year} {2018})},\
  \Eprint {http://arxiv.org/abs/1807.07975} {arXiv:1807.07975 [quant-ph]}
  \BibitemShut {NoStop}%
\bibitem [{\citenamefont {{Carignan-Dugas}}()}]{Dugas2018gcfgen}%
  \BibitemOpen
  \bibfield  {author} {\bibinfo {author} {\bibfnamefont {A.}~\bibnamefont
  {{Carignan-Dugas}}},\ }\href@noop {} {\bibinfo  {journal} {Extended analysis
  of RB-type protocols under gate-dependent noise models (Upcoming work)}\
  }\BibitemShut {NoStop}%
\bibitem [{\citenamefont {{Xue}}\ \emph {et~al.}(2018)\citenamefont {{Xue}},
  \citenamefont {{Watson}}, \citenamefont {{Helsen}}, \citenamefont {{Ward}},
  \citenamefont {{Savage}}, \citenamefont {{Lagally}}, \citenamefont
  {{Coppersmith}}, \citenamefont {{Eriksson}}, \citenamefont {{Wehner}},\ and\
  \citenamefont {{Vandersypen}}}]{Xue2018}%
  \BibitemOpen
\bibfield  {journal} {  }\bibfield  {author} {\bibinfo {author} {\bibfnamefont
  {X.}~\bibnamefont {{Xue}}}, \bibinfo {author} {\bibfnamefont {T.~F.}\
  \bibnamefont {{Watson}}}, \bibinfo {author} {\bibfnamefont {J.}~\bibnamefont
  {{Helsen}}}, \bibinfo {author} {\bibfnamefont {D.~R.}\ \bibnamefont
  {{Ward}}}, \bibinfo {author} {\bibfnamefont {D.~E.}\ \bibnamefont
  {{Savage}}}, \bibinfo {author} {\bibfnamefont {M.~G.}\ \bibnamefont
  {{Lagally}}}, \bibinfo {author} {\bibfnamefont {S.~N.}\ \bibnamefont
  {{Coppersmith}}}, \bibinfo {author} {\bibfnamefont {M.~A.}\ \bibnamefont
  {{Eriksson}}}, \bibinfo {author} {\bibfnamefont {S.}~\bibnamefont
  {{Wehner}}}, \ and\ \bibinfo {author} {\bibfnamefont {L.~M.~K.}\ \bibnamefont
  {{Vandersypen}}},\ }\href@noop {} {\bibfield  {journal} {\bibinfo  {journal}
  {arXiv e-prints}\ ,\ \bibinfo {eid} {arXiv:1811.04002}} (\bibinfo {year}
  {2018})},\ \Eprint {http://arxiv.org/abs/1811.04002} {arXiv:1811.04002
  [quant-ph]} \BibitemShut {NoStop}%
\bibitem [{\citenamefont {{Yang}}\ \emph {et~al.}(2018)\citenamefont {{Yang}},
  \citenamefont {{Chan}}, \citenamefont {{Harper}}, \citenamefont {{Huang}},
  \citenamefont {{Evans}}, \citenamefont {{Hwang}}, \citenamefont {{Hensen}},
  \citenamefont {{Laucht}}, \citenamefont {{Tanttu}}, \citenamefont {{Hudson}},
  \citenamefont {{Flammia}}, \citenamefont {{Itoh}}, \citenamefont {{Morello}},
  \citenamefont {{Bartlett}},\ and\ \citenamefont {{Dzurak}}}]{Yang2018}%
  \BibitemOpen
  \bibfield  {author} {\bibinfo {author} {\bibfnamefont {C.~H.}\ \bibnamefont
  {{Yang}}}, \bibinfo {author} {\bibfnamefont {K.~W.}\ \bibnamefont {{Chan}}},
  \bibinfo {author} {\bibfnamefont {R.}~\bibnamefont {{Harper}}}, \bibinfo
  {author} {\bibfnamefont {W.}~\bibnamefont {{Huang}}}, \bibinfo {author}
  {\bibfnamefont {T.}~\bibnamefont {{Evans}}}, \bibinfo {author} {\bibfnamefont
  {J.~C.~C.}\ \bibnamefont {{Hwang}}}, \bibinfo {author} {\bibfnamefont
  {B.}~\bibnamefont {{Hensen}}}, \bibinfo {author} {\bibfnamefont
  {A.}~\bibnamefont {{Laucht}}}, \bibinfo {author} {\bibfnamefont
  {T.}~\bibnamefont {{Tanttu}}}, \bibinfo {author} {\bibfnamefont {F.~E.}\
  \bibnamefont {{Hudson}}}, \bibinfo {author} {\bibfnamefont {S.~T.}\
  \bibnamefont {{Flammia}}}, \bibinfo {author} {\bibfnamefont {K.~M.}\
  \bibnamefont {{Itoh}}}, \bibinfo {author} {\bibfnamefont {A.}~\bibnamefont
  {{Morello}}}, \bibinfo {author} {\bibfnamefont {S.~D.}\ \bibnamefont
  {{Bartlett}}}, \ and\ \bibinfo {author} {\bibfnamefont {A.~S.}\ \bibnamefont
  {{Dzurak}}},\ }\href@noop {} {\bibfield  {journal} {\bibinfo  {journal}
  {arXiv e-prints}\ ,\ \bibinfo {eid} {arXiv:1807.09500}} (\bibinfo {year}
  {2018})},\ \Eprint {http://arxiv.org/abs/1807.09500} {arXiv:1807.09500
  [cond-mat.mes-hall]} \BibitemShut {NoStop}%
\bibitem [{\citenamefont {Kimmel}\ \emph {et~al.}(2014)\citenamefont {Kimmel},
  \citenamefont {da~Silva}, \citenamefont {Ryan}, \citenamefont {Johnson},\
  and\ \citenamefont {Ohki}}]{Kimmel2014}%
  \BibitemOpen
  \bibfield  {author} {\bibinfo {author} {\bibfnamefont {S.}~\bibnamefont
  {Kimmel}}, \bibinfo {author} {\bibfnamefont {M.~P.}\ \bibnamefont
  {da~Silva}}, \bibinfo {author} {\bibfnamefont {C.~a.}\ \bibnamefont {Ryan}},
  \bibinfo {author} {\bibfnamefont {B.~R.}\ \bibnamefont {Johnson}}, \ and\
  \bibinfo {author} {\bibfnamefont {T.~a.}\ \bibnamefont {Ohki}},\ }\href
  {\doibase 10.1103/PhysRevX.4.011050} {\bibfield  {journal} {\bibinfo
  {journal} {Physical Review X}\ }\textbf {\bibinfo {volume} {4}},\ \bibinfo
  {pages} {011050} (\bibinfo {year} {2014})}\BibitemShut {NoStop}%
\bibitem [{\citenamefont {{P{\'e}rez-Garc{\'{\i}}a}}\ \emph
  {et~al.}(2006)\citenamefont {{P{\'e}rez-Garc{\'{\i}}a}}, \citenamefont
  {{Wolf}}, \citenamefont {{Petz}},\ and\ \citenamefont
  {{Ruskai}}}]{GarciaPerez2006}%
  \BibitemOpen
  \bibfield  {author} {\bibinfo {author} {\bibfnamefont {D.}~\bibnamefont
  {{P{\'e}rez-Garc{\'{\i}}a}}}, \bibinfo {author} {\bibfnamefont {M.~M.}\
  \bibnamefont {{Wolf}}}, \bibinfo {author} {\bibfnamefont {D.}~\bibnamefont
  {{Petz}}}, \ and\ \bibinfo {author} {\bibfnamefont {M.~B.}\ \bibnamefont
  {{Ruskai}}},\ }\href {\doibase 10.1063/1.2218675} {\bibfield  {journal}
  {\bibinfo  {journal} {Journal of Mathematical Physics}\ }\textbf {\bibinfo
  {volume} {47}},\ \bibinfo {pages} {083506} (\bibinfo {year} {2006})},\
  \Eprint {http://arxiv.org/abs/math-ph/0601063} {math-ph/0601063} \BibitemShut
  {NoStop}%
\bibitem [{\citenamefont {Bhatia}(1997)}]{Bhatia1997}%
  \BibitemOpen
  \bibfield  {author} {\bibinfo {author} {\bibfnamefont {R.}~\bibnamefont
  {Bhatia}},\ }\href@noop {} {\emph {\bibinfo {title} {{Matrix Analysis}}}}\
  (\bibinfo  {publisher} {Springer},\ \bibinfo {year} {1997})\BibitemShut
  {NoStop}%
\end{thebibliography}%
\end{document}